\documentclass[11pt]{article}

\newcommand{\cmp}{Comm. Math. Phys.~}

\newcommand{\jpa}{J. Phys. A~}

\newcommand{\natphy}{Nature Phys.~}

\newcommand{\njp}{New. J. Phys.~}

\newcommand{\prl}{Phys. Rev. Lett.~}
\newcommand{\pra}{Phys. Rev. A~}

\newcommand{\pla}{Phys. Lett. A~}

\newcommand{\sci}{Science}

\usepackage{epstopdf}
\usepackage{subfigure}
\usepackage[sc]{mathpazo}
\usepackage{amsmath}
\usepackage{amssymb}
\usepackage{graphicx,color}
\usepackage{multirow}
\usepackage{enumerate}
\usepackage{amsthm}
\usepackage{amsfonts,mathrsfs}
\usepackage{geometry} 
\usepackage{ifthen}
\usepackage{cite}
\usepackage[unicode=true,pdfusetitle, bookmarks=true,bookmarksnumbered=false,bookmarksopen=false, breaklinks=false,pdfborder={0 0 0},backref=false,colorlinks=false] {hyperref}
\hypersetup{
colorlinks,linkcolor=myurlcolor,citecolor=myurlcolor,urlcolor=myurlcolor}
\definecolor{myurlcolor}{rgb}{0,0,0.7}

\newcommand{\blue}{\textcolor{blue}}

\newcommand{\proj}[1]{| #1\rangle\!\langle #1 |}

\usepackage{hyperref}
\hypersetup{pdfpagemode=UseNone}

\newcommand{\tinyspace}{\mspace{1mu}}

\newcommand{\op}[1]{\operatorname{#1}}

\newcommand{\abs}[1]{\left\lvert\tinyspace #1 \tinyspace\right\rvert}

\newcommand{\norm}[1]{\left\lVert\tinyspace #1 \tinyspace\right\rVert}

\renewcommand{\det}{\operatorname{det}}
\renewcommand{\t}{{\scriptscriptstyle\mathsf{T}}}

\newcommand{\setft}[1]{\mathrm{#1}}

\newcommand{\density}[1]{\setft{D}\left(#1\right)}

\newcommand{\rank}{\op{rank}}

\newcommand{\supp}{{\operatorname{supp}}} \newcommand{\var}{\operatorname{Var}}

\def \dif {\mathrm{d}}
\def \diag {\mathrm{diag}}

\def\complex{\mathbb{C}}
\def\real{\mathbb{R}}

\def\I{\mathbb{1}}

\newcommand{\Inner}[2]{\left\langle #1 , #2\right\rangle}
\newcommand{\Innerm}[3]{\left\langle #1 \left| #2 \right| #3 \right\rangle}

\newcommand{\Pa}[1]{\left(#1\right)}

\newcommand{\Br}[1]{\left[#1\right]}
\newcommand{\set}[1]{\{#1\}}
\newcommand{\Set}[1]{\left\{#1\right\}}

\newcommand{\bra}[1]{\langle#1|}

\newcommand{\ket}[1]{|#1\rangle}

\DeclareMathOperator{\trace}{Tr}

\newcommand{\Ptr}[2]{\trace_{#1}\Pa{#2}}

\newcommand{\Tr}[1]{\Ptr{}{#1}}

\def\cU{\mathcal{U}}

\def\bsA{\boldsymbol{A}}\def\bsB{\boldsymbol{B}}\def\bsC{\boldsymbol{C}}\def\bsD{\boldsymbol{D}}

\def\bsO{\boldsymbol{O}}
\def\bsQ{\boldsymbol{Q}}\def\bsT{\boldsymbol{T}}
\def\bsU{\boldsymbol{U}}\def\bsX{\boldsymbol{X}}

\def\bsa{\boldsymbol{a}}\def\bsb{\boldsymbol{b}}\def\bsc{\boldsymbol{c}}

\def\bsp{\boldsymbol{p}}
\def\bsu{\boldsymbol{u}}\def\bsx{\boldsymbol{x}}\def\bsy{\boldsymbol{y}}
\def\bsz{\boldsymbol{z}}

\def\rD{\mathrm{D}}

\def\sL{\mathscr{L}}\def\sM{\mathscr{M}}
\def\sP{\mathscr{P}}

\def\O{\textsf{O}}

\def\U{\textsf{U}}

\def\mC{\mathbb{C}}

\def\mR{\mathbb{R}}

\newtheorem{thrm}{Theorem}

\newtheorem{prop}{Proposition}
\newtheorem{cor}{Corollary}
\theoremstyle{definition}
\newtheorem{definition}{Definition}
\newtheorem{remark}{Remark}
\newtheorem{exam}{Example}

\begin{document}

\title{\bf\large Probability density functions of quantum mechanical observable uncertainties}

\author{\blue{Lin Zhang$^1$}\footnote{E-mail: godyalin@163.com},\ \blue{Jinping Huang$^1$},\ \blue{Jiamei Wang$^2$}\footnote{E-mail: wangjm@ahut.edu.cn},\ \blue{Shao-Ming Fei$^3$}\\
  {\it\small $^1$Institute of Mathematics, Hangzhou Dianzi University, Hangzhou 310018, China}\\
  {\it\small $^2$Department of Mathematics, Anhui University of Technology, Ma Anshan 243032,
China}\\
{\it\small $^3$School of Mathematical Sciences, Capital Normal
University, Beijing 100048, China}}
\date{}
\maketitle

\begin{abstract}
We study the uncertainties of quantum mechanical observables, quantified by the standard
deviation (square root of variance)  in Haar-distributed random pure states.
We derive analytically the probability density functions (PDFs) of the uncertainties of arbitrary qubit
observables. Based on these PDFs, the uncertainty
regions of the observables are characterized by the supports of the PDFs.
The state-independent uncertainty relations are then transformed into the
optimization problems over uncertainty regions, which opens a new
vista for studying state independent uncertainty relations.
Our results may be generalized to multiple observable case in higher dimensional spaces.
\end{abstract}

\newpage
\section{Introduction}

The uncertainty principle rules out the possibility to obtain
precise measurement outcomes simultaneously when one measures two
incomparable observables at the same time. Since the uncertainty
relation satisfied by the position and momentum
\cite{Heisenberg1927}, various uncertainty relations have been
extensively investigated
\cite{Dammeier2015njp,Li2015,Guise2018pra,Giorda2019pra,Xiao2019pra,Sponar2020pra}.
On the occasion of celebrating the 125th anniversary of the academic
journal "Science", the magazine listed 125 challenging scientific
problems \cite{Seife2005}. The 21st problem asks: Do deeper
principles underlie quantum uncertainty and nonlocality? As
uncertainty relations play significant roles in entanglement
detection
\cite{Hofman2003pra,Guhne2004prl,Guhne2009pra,Schwonnek2017prl,Qian2018qip,Zhao2019prl}
and quantum nonlocality \cite{Oppenheim2010}, and many others, it is
desirable to explore the mathematical structures and physical
implications of uncertainties in more details from various
perspectives.

The state-dependent Robertson-Schr\"{o}dinger uncertainty relation
\cite{Kennard1927,Weyl1928,Robertson1929,Schrodinger1930} is of form:
\begin{eqnarray*}
(\Delta_\rho\bsA)^2(\Delta_\rho\bsB)^2\geqslant\frac14\Br{
(\langle\set{\bsA,\bsB}\rangle_\rho -
\langle\bsA\rangle_\rho\langle\bsB\rangle_\rho)^2+\langle[\bsA,\bsB]\rangle_\rho^2},
\end{eqnarray*}
where $\set{\bsA,\bsB}:=\bsA\bsB+\bsB\bsA$,
$[\bsA,\bsB]:=\bsA\bsB-\bsB\bsA$, and
$(\Delta_\rho\bsX)^2:=\Tr{\bsX^2\rho}-\Tr{\bsX\rho}^2$ is the
variance of $\bsX$ with respect to the state $\rho$, $\bsX=\bsA,\bsB$.

Recently, state-independent uncertainty relations have been investigated \cite{Guhne2004prl,Schwonnek2017prl},
which have direct applications to entanglement detection.
In order to get state-independent uncertainty relations, one considers the sum
of the variances and solves the following optimization problems:
\begin{eqnarray}
\var_\rho(\bsA)+\var_\rho(\bsB)&\geqslant&
\min_{\rho\in\rD(\mC^d)}\Pa{\var_\rho(\bsA)+\var_\rho(\bsB)},\label{eq:1}\\
\Delta_\rho\bsA+\Delta_\rho\bsB&\geqslant&
\min_{\rho\in\rD(\mC^d)}\Pa{\Delta_\rho\bsA+\Delta_\rho\bsB},\label{eq:2}
\end{eqnarray}
where $\var_\rho(\bsX)=(\Delta_\rho\bsX)^2$ is the variance of the
observable $\bsX$ associated to state $\rho\in\rD(\mC^d)$.

Efforts have been devoted to
provide quantitative uncertainty bounds for the above inequalities \cite{Busch2019}.
However, searching for such uncertainty bounds may be not the best way to
get new uncertainty relations \cite{Zhang2018}. Recently, Busch
and Reardon-Smitha proposed to consider the \emph{uncertainty
region} \cite{Busch2019} of two observables $\bsA$ and $\bsB$,
instead of finding the bounds based on some particular choice of uncertainty
functional, typically such as the product or sum of uncertainties
\cite{Maccone2014}. Once we can identify what the structures of
uncertainty regions are, we can infer specific information about the
state with the minimal uncertainty in some sense.
In view of this, the above two optimization problems
\eqref{eq:1} and \eqref{eq:2} become
\begin{eqnarray*}
\min_{\rho\in\rD(\mC^d)}\Pa{\var_\rho(\bsA)+\var_\rho(\bsB)} &=&
\min\Set{x^2+y^2:(x,y)\in\cU^{(\text{m})}_{\Delta\bsA,\Delta\bsB}},\\
\min_{\rho\in\rD(\mC^d)}\Pa{\Delta_\rho\bsA+\Delta_\rho\bsB} &=&
\min\Set{x+y:(x,y)\in\cU^{(\text{m})}_{\Delta\bsA,\Delta\bsB}},
\end{eqnarray*}
where $\cU^{(\text{m})}_{\Delta\bsA,\Delta\bsB}$ is the so-called
uncertainty region of two observables $\bsA$ and $\bsB$ defined by
\begin{eqnarray*}
\cU^{(\text{m})}_{\Delta\bsA,\Delta\bsB} =
\Set{(\Delta_\rho\bsA,\Delta_\rho\bsB)\in\mR^2_+:
\rho\in\rD(\mC^d)}.
\end{eqnarray*}

Random matrix theory or probability theory are powerful tools in
quantum information theory. Recently, the non-additivity of quantum
channel capacity \cite{Hastings2009} has been cracked via probabilistic
tools. The Duistermaat-Heckman measure on moment polytope has been used to
derive the probability distribution density of one-body quantum
marginal states of multipartite random quantum states
\cite{Christandl2014,Dartois2020} and that of classical probability
mixture of random quantum states \cite{Zhang2018pla,Zhang2019jpa}.
As a function of random quantum pure states, the probability density
function (PDF) of quantum expectation value of an observable is also
analytical calculated \cite{Venuti2013}. Motivated by these works,
we investigate the joint probability density functions
of uncertainties of observables. By doing so, we find that it is
not necessarily to solve directly the uncertainty regions of
observables. It is sufficient to identify the support of such PDF
because PDF vanishes exactly beyond the uncertainty regions. Thus
all the problems are reduced to compute the PDF of uncertainties of
observables, since all information concerning uncertainty regions and
state-independent uncertainty relations are encoded in such PDFs.
In \cite{Zhang2021preprint} we have studied such PDFs for the random mixed quantum
state ensembles, where all problems concerning qubit observables are completely solved, i.e., analytical formulae of the PDFs of uncertainties are obtained, and the characterization of
uncertainty regions over which the optimization problems for
state-independent lower bound of sum of variances is presented. In this
paper, we will focus the same problem for random pure quantum
state ensembles.

Let $\delta(x)$ be delta function \cite{Hoskins2009} defined by
\begin{eqnarray*}
\delta(x)=\begin{cases} +\infty,&\text{if }x\neq0;\\0,&\text{if
}x=0.\end{cases}
\end{eqnarray*}
One has $\Inner{\delta}{f}:=\int_\real f(x)\delta(x)\dif
x=f(0)$. Denote by $\delta_a(x):=\delta(x-a)$. Then
$\Inner{\delta_a}{f} = f(a)$.
Let $Z(g):=\Set{x\in D(g):g(x)=0}$ be the zero set of function $g(x)$ with its domain $D(g)$. We will use the following definition.

\begin{definition}[\cite{lz2020ijtp,Zuber2020}] If $g:\real\to\real$ is a smooth function
(the first derivative $g'$ is a continuous function) such that
$Z(g)\cap Z(g')=\emptyset$, then the composite $\delta\circ g$ is
defined by:
\begin{eqnarray*}
\delta(g(x)) = \sum_{x\in Z(g)} \frac1{\abs{g'(x)}}\delta_x.
\end{eqnarray*}
\end{definition}

\section{Uncertainty regions of observables}

We can extend the notion of the uncertainty region of two
observables $\bsA$ and $\bsB$, put forward in \cite{Busch2019}, into
that of multiple observables.

\begin{definition}
Let $(\bsA_1,\ldots,\bsA_n)$ be an $n$-tuple of qudit observables
acting on $\mC^d$. The \emph{uncertainty region} of such $n$-tuple
$(\bsA_1,\ldots,\bsA_n)$, for the mixed quantum state ensemble, is
defined by
\begin{eqnarray*}
\cU^{(\text{m})}_{\Delta\bsA_1,\ldots,\Delta\bsA_n}:=\Set{(\Delta_\rho\bsA_1,\ldots,\Delta_\rho\bsA_n)\in\mR^n_+:\rho\in\rD(\mC^d)}.
\end{eqnarray*}
Similarly, the \emph{uncertainty region} of such $n$-tuple
$(\bsA_1,\ldots,\bsA_n)$, for the pure quantum state ensemble, is
defined by
\begin{eqnarray*}
\cU^{(\text{p})}_{\Delta\bsA_1,\ldots,\Delta\bsA_n}:=\Set{(\Delta_\psi\bsA_1,\ldots,\Delta_\psi\bsA_n)\in\mR^n_+:\ket{\psi}\in\mC^d}.
\end{eqnarray*}
Apparently,
$\cU^{(\text{p})}_{\Delta\bsA_1,\ldots,\Delta\bsA_n}\subset\cU^{(\text{m})}_{\Delta\bsA_1,\ldots,\Delta\bsA_n}$.
\end{definition}
Note that our definition about uncertainty region is different from
the one given in \cite{Dammeier2015njp}. In the above definition, we
use the standard deviation instead of variance.

Next we will show that
$\cU^{(\text{m})}_{\Delta\bsA_1,\ldots,\Delta\bsA_n}$ is contained
in the supercube in $\mR^n_+$. To this end, we study the following
sets $\sP(\bsA)=\Set{\text{Var}_\psi(\bsA): \ket{\psi}\in\mC^d}$ and
$\sM(\bsA)=\Set{\text{Var}_\rho(\bsA): \rho\in\density{\mC^d}}$ for
a qudit observable $\bsA$ acting on $\mC^d$. The relationship
between both sets $\sP(\bsA)$ and $\sM(\bsA)$ is summarized into the
following proposition.
\begin{prop}\label{prop:convx}
It holds that
\begin{eqnarray*}
\sP(\bsA)=\sM(\bsA)=\mathrm{conv}(\sP(\bsA))
\end{eqnarray*}
is a closed interval $[0,\max_\psi\mathrm{Var}_\psi(\bsA)]$.
\end{prop}

\begin{proof}
Note that
$\sP(\bsA)\subset\sM(\bsA)\subset\mathrm{conv}(\sP(\bsA))$. Here the
first inclusion is apparently; the second inclusion follows
immediately from the result obtained in \cite{Petz2012}: For any
density matrix $\rho\in\density{\mC^d}$ and a qudit observable
$\bsA$, there is a pure state ensemble decomposition $\rho=\sum_j
p_j\proj{\psi_j}$ such that
\begin{eqnarray}\label{eq:vardecom}
\text{Var}_\rho(\bsA) = \sum_j p_j\text{Var}_{\psi_j}(\bsA).
\end{eqnarray}
Since all pure states on $\mC^d$ can be generated via a fixed
$\psi_0$ and the whole unitary group $\U(d)$, it follows that
\begin{eqnarray*}
\sP(\bsA) = \mathrm{im}(\Phi),
\end{eqnarray*}
where the mapping $\Phi:\U(d)\to \sP(\bsA)$ is defined by
$\Phi(\bsU)=\mathrm{Var}_{\bsU\psi_0}(\bsA)$. This mapping $\Phi$ is
surjective and continuous. Due to the fact that $\U(d)$ is a compact
Lie group, we see that $\Phi$ can attain maximal and minimal values
over the unitary group $\U(d)$. In fact, $\min_{\U(d)}\Phi=0$. This
can be seen if we take some $\bsU$ such that $\bsU\ket{\psi_0}$ is
an eigenvector of $\bsA$. Since $\U(d)$ is also connected, then
$\mathrm{im}(\Phi)=\Phi(\U(d))$ is also connected, thus
$\mathrm{im}(\Phi)=[0,\max_{\U(d)}\Phi]$. This amounts to say that
$\sP(\bsA)$ is a closed interval $[0,\max_{\U(d)}\Phi]$ which means
that $\sP(\bsA)$ is a compact and convex set, i.e.,
\begin{eqnarray*}
\sP(\bsA) =\mathrm{conv}(\sP(\bsA)).
\end{eqnarray*}
Therefore
\begin{eqnarray*}
\sP(\bsA)
=\sM(\bsA)=\mathrm{conv}(\sP(\bsA))=[0,\max_{\U(d)}\Phi]=[0,\max_{\psi}\mathrm{Var}_\psi(\bsA)].
\end{eqnarray*}
This completes the proof.
\end{proof}

Next, we determine
$\max_{\rho\in\density{\complex^d}}\var_\rho(\bsA)$ for an
observable $\bsA$. To this end, we recall the following
$(d-1)$-dimensional probability simplex, which is defined by
\begin{eqnarray*}
\Delta_{d-1}:=\Set{\bsp=(p_1,\ldots,p_d)\in\real^d:p_k\geqslant0(\forall
k\in[d]),\sum_jp_j=1}.
\end{eqnarray*}
Its interior of $\Delta_{d-1}$ is denoted by $\Delta^\circ_{d-1}$:
\begin{eqnarray*}
\Delta^\circ_{d-1}:=\Set{\bsp=(p_1,\ldots,p_d)\in\real^d:p_k>0(\forall
k\in[d]),\sum_jp_j=1}.
\end{eqnarray*}
This indicates that a point $\bsx$ in the boundary
$\partial\Delta_{d-1}$ means that there must be at least a component
$x_i=0$ for some $i\in[d]$. Now we separate the boundary of
$\partial\Delta_{d-1}$ into the union of the following subsets:
\begin{eqnarray*}
\partial\Delta_{d-1} =\bigcup^{d}_{j=1} F_j,
\end{eqnarray*}
where $F_j:=\Set{\bsx\in\partial\Delta_{d-1}: x_j=0}$. Although the
following result is known in 1935 \cite{Bhatia2000}, we still
include our proof for completeness.
\begin{prop}\label{prop:varmax}
Assume that $\bsA$ is an observable acting on $\complex^d$. Denote
the vector consisting of eigenvalues of $\bsA$ by $\lambda(\bsA)$
with components being $\lambda_1(\bsA)\leqslant\cdots\leqslant
\lambda_d(\bsA)$. It holds that
\begin{eqnarray*}
\max\Set{\var_\rho(\bsA):\rho\in\density{\complex^d}}
=\frac14\Pa{\lambda_{\max}(\bsA)-\lambda_{\min}(\bsA)}^2.
\end{eqnarray*}
Here $\lambda_{\min}(\bsA)=\lambda_1(\bsA)$ and
$\lambda_{\max}(\bsA)=\lambda_d(\bsA)$.
\end{prop}

\begin{proof}
Assume that $\bsa:=\lambda(\bsA)$ where $a_j:=\lambda_j(\bsA)$. Note
that
$\var_\rho(\bsA)=\Tr{\bsA^2\rho}-\Tr{\bsA\rho}^2=\Inner{\bsa^2}{\bsD_{\bsU}\lambda(\rho)}-\Inner{\bsa}{\bsD_{\bsU}\lambda(\rho)}^2$,
where $\bsa=(a_1,\ldots,a_d)^\t, \bsa^2=(a^2_1,\ldots,a^2_d)^\t$,
and $\bsD_{\bsU}=\overline{\bsU}\circ\bsU$ (here $\circ$ stands for
Schur product, i.e., entrywise product), and
$\lambda(\rho)=(\lambda_1(\rho),\ldots,\lambda_d(\rho))^\t$. Denote
\begin{eqnarray*}
\bsx:=\bsD_{\bsU}\lambda(\rho)\in\Delta_{d-1}:=\Set{\bsp=(p_1,\ldots,p_d)\in\real^d_+:\sum_jp_j=1},
\end{eqnarray*}
the $(d-1)$-dimensional probability simplex. Then
\begin{eqnarray*}
\var_\rho(\bsA)=\Inner{\bsa^2}{\bsx}-\Inner{\bsa}{\bsx}^2=\sum^d_{j=1}a^2_jx_j
- \Pa{\sum^d_{j=1}a_jx_j}^2=:f(\bsx).
\end{eqnarray*}

(i) If $d=2$,
\begin{eqnarray*}
f(x_1,x_2)&=&a_1^2x_1+a_2^2x_2-(a_1x_1+a_2x_2)^2\notag\\
&=&a_1^2x_1+a_2^2(1-x_1)-(a_1x_1+a_2(1-x_2))^2\notag\\
&=&(a_2-a_1)^2\Br{\frac14-\Pa{x_1-\frac12}^2}\leqslant
\frac14(a_2-a_1)^2,
\end{eqnarray*}
implying that $f_{\max}=\frac14(a_2-a_1)^2$ when $x_1=x_2=\frac12$.

(ii) If $d\geqslant3$, without loss of generality, we assume that
$a_1< a_2<\cdots<a_d$, we will show that the function $f$ takes its
maximal value on the point
$(x_1,x_2,\ldots,x_{d-1},x_d)=(\tfrac12,0,\ldots,0,\tfrac12)$, with
the maximal value being
$\frac14(a_d-a_1)^2=\frac14\Pa{\lambda_{\max}(\bsA)-\lambda_{\min}(\bsA)}^2$.
Then, using Lagrangian multiplier method, we let
\begin{eqnarray*}
L(x_1,x_2,\cdots,x_d,\lambda)=\sum_{i=1}^d a_i^2x_i-\Pa{\sum_{i=1}^d
a_ix_i}^2+\lambda\Pa{\sum_{i=1}^d x_i-1}.
\end{eqnarray*}
Thus
\begin{eqnarray}
\frac{\partial L}{\partial x_i}&=&a_i^2-2 a_i\Pa{\sum_{i=1}^d
a_ix_i}+\lambda=0\quad(i=1,\ldots,d),\label{L-criticalpoints}
\\
\frac{\partial L}{\partial \lambda}&=&\sum_{i=1}^d x_i-1=0.\notag
\end{eqnarray}
Denote $m:=\sum_{i=1}^d a_ix_i$. Because \eqref{L-criticalpoints}
holds for all $i=1,\ldots,d$, we see that
\begin{eqnarray*}
\lambda=-a_i^2+2a_im=-a_j^2+2a_jm,
\end{eqnarray*}
that is, $m=\frac{a_i+a_j}2$ and $\lambda=a_ia_j$ for all distinct
indices $i$ and $j$. Furthermore, for all distinct indices $i$ and
$j$, the system of equations $\sum_{i=1}^d a_ix_i=m=\frac{a_i+a_j}2$
have no solution on $\Delta_{d-1}$. Hence there is no stationary
point on $\Delta_{d-1}$, and thus $f_{\max}$ is obtained on the
boundary $\partial\Delta_{d-1}$ of $\Delta_{d-1}$. Suppose, by
induction, that the conclusion holds for the case where
$d=k\geqslant2$, i.e., the function $f$ takes its maximal value
$f_{\max}=\frac14(a_k-a_1)^2$ on the point
$(x_1,x_2,\ldots,x_{k-1},x_k)=(\tfrac12,0,\ldots,0,\tfrac12)\in\partial\Delta_{k-1}$.

Next we consider the case where $d=k+1$, i.e., the extremal value of
$f(x_1,\ldots,x_{k+1})$ on $\partial\Delta_k$. If $\bsx\in
F_j\subset\Delta_k$, where $j\in\set{2,\ldots,k}$, then
$f_{k+1}(x_1,x_2,\cdots,x_{k+1})
=f_k(y_1,y_2,\cdots,y_{k})=\sum_{i=1}^{k} b_i^2y_i-(\sum_{i=1}^{k}
b_iy_i)^2,$ where $b_1=a_1,\cdots, b_{i-1}=a_{i-1}, b_i=a_{i+1},
\cdots,b_k=a_{k+1},$ it is obvious that $b_1<b_2<\cdots<b_k$. By the
previous assumption, we have
\begin{eqnarray*}
f_{\max}=\frac14(b_k-b_1)^2=\frac14(a_{k+1}-a_1)^2,
\end{eqnarray*}
and the maximal value is attained at
$(x_1,x_2,\ldots,x_{k},x_{k+1})=(\tfrac12,0,\ldots,0,\tfrac12)$;
similarly $f_{\max}=\frac14(a_{k+1}-a_2)^2$ is attained on $F_1$;
$f_{\max}=\frac14(a_k-a_1)^2$ is attained on $F_{k+1}$.

By comparing these extremal values, we know that
\begin{eqnarray*}
f_{\max}=\frac14(a_{k+1}-a_1)^2
\end{eqnarray*}
is attained on $\partial\Delta_k$ and the maximal value is attained
at $(x_1,x_2,\ldots,x_{k},x_{k+1})=(\tfrac12,0,\ldots,0,\tfrac12)$.
\end{proof}

\begin{remark}
By using spectral decomposition theorem to $\bsA$, we get that
$\bsA=\sum^d_{j=1} a_j\proj{a_j}$. Denote
$\ket{\psi}=\frac{\ket{a_1}+\ket{a_d}}{\sqrt{2}}$. Then we see that
\begin{eqnarray*}
\var_\psi(\bsA) = \Tr{\bsA^2\proj{\psi}} -  \Tr{\bsA\proj{\psi}}^2 =
\frac14(\lambda_{\max}(\bsA)-\lambda_{\min}(\bsA))^2.
\end{eqnarray*}
\end{remark}

\begin{prop}
Let $(\bsA_1,\ldots,\bsA_n)$ be an $n$-tuple of qudit observables
acting on $\mC^d$. Denote
$v(\bsA_k):=\frac12(\lambda_{\max}(\bsA_k)-\lambda_{\min}(\bsA_k))$,
where $k=1,\ldots,n$ and $\lambda_{\max/\min}(\bsA_k)$ stands for
the maximal/minimal eigenvalue of $\bsA$. Then
\begin{eqnarray*}
\cU^{(\mathrm{m})}_{\Delta\bsA_1,\ldots,\Delta\bsA_n}\subset
\Br{0,v(\bsA_1)}\times\cdots\times \Br{0,v(\bsA_n)}.
\end{eqnarray*}
\end{prop}

\begin{proof}
The proof is easily obtained by combining
Proposition~\ref{prop:convx} and Proposition~\ref{prop:varmax}.
\end{proof}

\section{PDFs of expectation values and uncertainties of qubit observables}\label{app:A}

Assume $\bsA$ is a non-degenerate positive matrix with eigenvalues
$\lambda_1(\bsA)<\cdots<\lambda_d(\bsA)$. Denote by
$\lambda(\bsA)=(\lambda_1(\bsA),\ldots,\lambda_d(\bsA))$. In view of
the speciality of pure state ensemble and noting
Propositions~\ref{prop:convx} and \ref{prop:varmax}, we will
consider only the variances of observable $\bsA$ over pure states.
In fact, the same problem is also considered very recently for mixed
state \cite{Zhang2021preprint}. Then the probability density
function of $\langle\bsA\rangle_\psi:=\Innerm{\psi}{\bsA}{\psi}$ is
defined by
\begin{eqnarray*}
f^{(d)}_{\langle\bsA\rangle}(r)
:=\int\delta(r-\langle\bsA\rangle_\psi)\dif\mu(\psi).
\end{eqnarray*}
Here $\dif\mu(\psi)$ is the so-called uniform probability measure,
which is invariant under the unitary rotations, can be realized by
the following way:
\begin{eqnarray*}
\dif\mu(\psi) =
\frac{\Gamma(d)}{2\pi^d}\delta(1-\norm{\psi})[\dif\psi],
\end{eqnarray*}
where $[\dif\psi]=\prod^d_{k=1}\dif x_k\dif y_k$ for
$\psi_k=x_k+\mathrm{i}y_k(k=1,\ldots,d)$, and $\Gamma(\cdot)$ is the
Gamma function. Thus
\begin{eqnarray*}
f^{(d)}_{\langle\bsA\rangle}(r)
=\Gamma(d)\int_{\real^d_+}\delta\Pa{r-\sum^d_{i=1}\lambda_i(\bsA)r_i}\delta\Pa{1-\sum^d_{i=1}r_i}\prod^d_{i=1}\dif
r_i.
\end{eqnarray*}
For completeness, we will give a different proof of it although the
following result is already obtained in \cite{Venuti2013}:

\begin{prop}\label{prop:expectn}
For a given quantum observable $\bsA$ with simple spectrum
$\lambda(\bsA)=(\lambda_1(\bsA),\ldots,\lambda_d(\bsA))$, where $
\lambda_1(\bsA)<\cdots<\lambda_d(\bsA)$, the probability density
function of $\langle\bsA\rangle_\psi$, where $\ket{\psi}$ a
Haar-distributed random pure state on $\complex^d$, is given by the
following:
\begin{eqnarray}\label{eq:qubitexp}
f^{(d)}_{\langle\bsA\rangle}(r)=(-1)^{d-1}(d-1)\sum^d_{i=1}\frac{(r-\lambda_i(\bsA))^{d-2}}{\prod_{j\in\hat
i}(\lambda_i(\bsA)-\lambda_j(\bsA))}H(r-\lambda_i(\bsA)),
\end{eqnarray}
where $\hat i:=\set{1,2,\ldots,d}\backslash\set{i}$ and $H$ is the
so-called Heaviside function, defined by $H(t)=1$ if $t>0$, $0$
otherwise. Thus the support of $f^{(d)}_{\langle\bsA\rangle}(r)$ is
the closed interval $[\lambda_1(\bsA),\lambda_n(\bsA)]$. In
particular, for $d=2$, we have
\begin{eqnarray*}
f^{(2)}_{\langle\bsA\rangle}(r)=\frac1{\lambda_2(\bsA)-\lambda_1(\bsA)}(H(r-\lambda_1(\bsA))-H(r-\lambda_2(\bsA))).
\end{eqnarray*}
\end{prop}

\begin{proof}
By performing Laplace transformation $(r\to s)$ of
$f^{(d)}_{\langle\bsA\rangle}(r)$, we get that
\begin{eqnarray*}
\sL(f^{(d)}_{\langle\bsA\rangle})(s) = \Gamma(d)\int
\exp\Pa{-s\sum^d_{i=1}\lambda_i(\bsA)r_i}\delta\Pa{1-\sum^d_{i=1}r_i}\prod^d_{i=1}\dif
r_i.
\end{eqnarray*}
Let
$$
F_s(t) :=\Gamma(d)\int
\exp\Pa{-s\sum^d_{i=1}\lambda_i(\bsA)r_i}\delta\Pa{t-\sum^d_{i=1}r_i}\prod^d_{i=1}\dif
r_i.
$$
Still by performing Laplace transformation $(t\to x)$ of $F_s(t)$:
\begin{eqnarray*}
\sL(F_s)(x) &=&\Gamma(d)\int
\exp\Pa{-s\sum^d_{i=1}\lambda_i(\bsA)r_i}\exp\Pa{-x\sum^d_{i=1}r_i}\prod^d_{i=1}\dif
r_i\\
&=&\Gamma(d)\prod^d_{i=1}\int^\infty_0 \exp\Pa{-(s\lambda_i(\bsA)+x)r_i}\dif r_i\\
&=&\frac{\Gamma(d)}{\prod^d_{i=1}(s\lambda_i(\bsA)+x)},
\end{eqnarray*}
implying that \cite{zhang2018qip}
\begin{eqnarray*}
F_s(t) =
\Gamma(d)\sum^d_{i=1}\frac{\exp\Pa{-\lambda_i(\bsA)st}}{(-s)^{d-1}\prod_{j\in\hat
i}(\lambda_i(\bsA)-\lambda_j(\bsA))},
\end{eqnarray*}
where $\hat i:=\Set{1,\ldots,d}\backslash\set{i}$. Thus
\begin{eqnarray*}
\sL(f^{(d)}_{\langle\bsA\rangle})(s)=F_s(1) =
\Gamma(d)\sum^d_{i=1}\frac{\exp\Pa{-\lambda_i(\bsA)s}}{(-s)^{d-1}\prod_{j\in\hat
i}(\lambda_i(\bsA)-\lambda_j(\bsA))}.
\end{eqnarray*}
Therefore, we get that
\begin{eqnarray*}
f^{(d)}_{\langle\bsA\rangle}(r)
=(-1)^{d-1}(d-1)\sum^d_{i=1}\frac{H(r-\lambda_i(\bsA))(r-\lambda_i(\bsA))^{d-2}}{\prod_{j\in\hat
i}(\lambda_i(\bsA)-\lambda_j(\bsA))},
\end{eqnarray*}
where $H(r-\lambda_i(\bsA))$ is the so-called Heaviside function,
defined by $H(t)=1$ if $t>0$; otherwise $0$. The support of this pdf
is the closed interval $[l,u]$ where
$$
l=\min\Set{\lambda_i(\bsA):i=1,\ldots,d},\quad
u=\max\Set{\lambda_i(\bsA):i=1,\ldots,d}.
$$
The normalization of $f^{(d)}_{\langle\bsA\rangle}(r)$ (i.e.,
$\int_\real f^{(d)}_{\langle\bsA\rangle}(r)\dif r=1$) can be checked
by assuming $\lambda_1<\lambda_2<\cdots<\lambda_d$, then
$[l,u]=[\lambda_1,\lambda_n]$ since
$f^{(d)}_{\langle\bsA\rangle}(r)$ is a symmetric of $\lambda_i$'s.
\end{proof}

\subsection{The case for one qubit observable}

Let us now turn to the qubit observables. Any qubit observable
$\bsA$, which may be parameterized as
\begin{equation}\label{AA}
\bsA=a_0\I+\bsa\cdot\boldsymbol{\sigma},\quad (a_0,\bsa)\in\real^4,
\end{equation}
where $\I$ is the identity matrix on the qubit Hilbert space
$\complex^2$, and $\boldsymbol{\sigma}=(\sigma_1,\sigma_2,\sigma_3)$
are the vector of the standard Pauli matrices:
\begin{eqnarray*}
\sigma_1=\Pa{\begin{array}{cc}
               0 & 1 \\
               1 & 0
             \end{array}
},\quad \sigma_2=\Pa{\begin{array}{cc}
               0 & -\mathrm{i} \\
               \mathrm{i} & 0
             \end{array}
},\quad \sigma_3=\Pa{\begin{array}{cc}
               1 & 0 \\
               0 & -1
             \end{array}
}.
\end{eqnarray*}
Without loss of generality, we assume that our qubit observables are
of simple eigenvalues, otherwise the problem is trivial. Thus the
two eigenvalues of $\bsA$ are
\begin{eqnarray*}
\lambda_k(\bsA)=a_0+(-1)^ka, \qquad k=1,2,
\end{eqnarray*}
with $a:=\abs{\bsa}=\sqrt{a_1^2+a_2^2+a_3^2}>0$ being the length of
vector $\bsa =(a_1,a_2,a_3)\in \real^3$. Thus \eqref{eq:qubitexp}
becomes
\begin{eqnarray*}
f^{(2)}_{\langle\bsA\rangle}(r) =
\frac1{\lambda_2(\bsA)-\lambda_1(\bsA)}[H(r-\lambda_1(\bsA))-H(r-\lambda_2(\bsA))].
\end{eqnarray*}

\begin{thrm}\label{th:vardis}
For the qubit observable $\bsA$ defined by Eq.~\eqref{AA}, the
probability density function of $\Delta_\psi \bsA$, where $\psi$ is
a Haar-distributed random pure state on $\complex^2$, is given by
\begin{eqnarray*}
 f^{(2)}_{\Delta\bsA}(x) = \frac x{\abs{\bsa}\sqrt{\abs{\bsa}^2-x^2}},\qquad
x\in[0,\abs{\bsa}).
\end{eqnarray*}
\end{thrm}

\begin{proof}
Note that
\begin{eqnarray}\label{eq:2nddelta}
\delta(r^2-r^2_0)=\frac1{2\abs{r_0}}\Br{\delta(r-r_0)+\delta(r+r_0)}.
\end{eqnarray}
For $x\geqslant0$, because
\begin{eqnarray*}
\delta(x^2-\Delta_\psi \bsA^2) = \frac1{2x}\Br{\delta(x+\Delta_\psi
\bsA)+\delta(x-\Delta_\psi \bsA)}=\frac1{2x}\delta(x-\Delta_\psi
\bsA),
\end{eqnarray*}
we see that
\begin{eqnarray*}
f^{(2)}_{\Delta\bsA}(x) = \int\delta(x-\Delta_\psi
\bsA)\dif\mu(\psi) = 2x\int\delta\Pa{x^2-\Delta^2_\psi
\bsA}\dif\mu(\psi).
\end{eqnarray*}
For any complex $2\times 2$ matrix $\bsA$,
$\bsA^2=\Tr{\bsA}\bsA-\det(\bsA)\I$. Then $\Delta^2_\psi
\bsA=(\langle\bsA\rangle_\psi-\lambda_1(\bsA))(\lambda_2(\bsA)-\langle\bsA\rangle_\psi)$
\begin{eqnarray*}
\delta\Pa{x^2-\Delta^2_\psi \bsA} &=&
\delta\Pa{x^2-(\langle\bsA\rangle_\psi-\lambda_1(\bsA))(\lambda_2(\bsA)-\langle\bsA\rangle_\psi)}.
\end{eqnarray*}
In particular, we see that
\begin{eqnarray*}
f^{(2)}_{\Delta\bsA}(x) &=&
2x\int^{\lambda_2(\bsA)}_{\lambda_1(\bsA)}
\dif r\delta\Pa{x^2-(r-\lambda_1(\bsA))(\lambda_2(\bsA)-r)}\int_{\complex^2}\delta(r-\langle\bsA\rangle_\psi)\dif\mu(\psi)\notag\\
&=& 2x\int^{\lambda_2(\bsA)}_{\lambda_1(\bsA)} \dif
rf^{(2)}_{\langle\bsA\rangle}(r)\delta\Pa{x^2-(r-\lambda_1(\bsA))(\lambda_2(\bsA)-r)}.
\end{eqnarray*}
Denote $f_x(r)=x^2-(r-\lambda_1(\bsA))(\lambda_2(\bsA)-r)$. Thus
$\partial_rf_x(r)=2r-\lambda_1(\bsA)-\lambda_2(\bsA)$. Then
$f_x(r)=0$ has two distinct roots in
$[\lambda_1(\bsA),\lambda_2(\bsA)]$ if and only if
$x\in\left[0,\frac{V_2(\lambda(\bsA))}2\right)$, where
$V_2(\lambda(\bsA))=\lambda_2(\bsA)-\lambda_1(\bsA)$. Now the roots
are given by
\begin{eqnarray*}
r_\pm(x)=\frac{\lambda_1(\bsA)+\lambda_2(\bsA)\pm\sqrt{V_2(\lambda(\bsA))^2-4x^2}}2.
\end{eqnarray*}
Thus
\begin{eqnarray*}
\delta\Pa{f_x(r)}=\frac1{\abs{\partial_{r=r_+(x)}
f_x(r)}}\delta_{r_+(x)}+\frac1{\abs{\partial_{r=r_-(x)}
f_x(r)}}\delta_{r_-(x)},
\end{eqnarray*}
implying that
\begin{eqnarray*}
f^{(2)}_{\Delta \bsA}(x) =
\frac{4x}{V_2(\lambda(\bsA))\sqrt{V_2(\lambda(\bsA))^2-4x^2}}.
\end{eqnarray*}

Now for $\bsA=a_0\I+\bsa\cdot\boldsymbol{\sigma}$, we have
$V_2(\lambda(\bsA))=2\abs{\bsa}$. Substituting this into the above
expression, we get the desired result:
\begin{eqnarray*}
f^{(2)}_{\Delta \bsA}(x) = \frac
x{\abs{\bsa}\sqrt{\abs{\bsa}^2-x^2}},
\end{eqnarray*}
where $x\in[0,\abs{\bsa})$. This is the desired result.
\end{proof}

\subsection{The case for two qubit
observables}

Let $\bsA=a_0+\bsa\cdot\boldsymbol{\sigma}$ and
$\bsB=b_0+\bsb\cdot\boldsymbol{\sigma}$
\begin{eqnarray*}
&& f^{(2)}_{\langle\bsA\rangle,\langle\bsB\rangle}(r,s) =
\int\delta(r-\Innerm{\psi}{\bsA}{\psi})\delta(s-\Innerm{\psi}{\bsB}{\psi})\dif\mu(\psi)\\
&&=\frac1{(2\pi)^2}\int_{\real^2}\dif\alpha\dif\beta
\exp\Pa{\mathrm{i}(r\alpha+s\beta)}\int
\exp\Pa{-\mathrm{i}\Innerm{\psi}{\alpha\bsA+\beta\bsB}{\psi}}\dif\mu(\psi),
\end{eqnarray*}
where
\begin{eqnarray*}
&&\int
\exp\Pa{-\mathrm{i}\Innerm{\psi}{\alpha\bsA+\beta\bsB}{\psi}}\dif\mu(\psi)
=
\int^{\lambda_+(\alpha\bsA+\beta\bsB)}_{\lambda_-(\alpha\bsA+\beta\bsB)}
\exp\Pa{-\mathrm{i}t}f_2(t)\dif t\\
&&=\frac1{2\abs{\alpha\bsa+\beta\bsb}}\int^{\lambda_+(\alpha\bsA+\beta\bsB)}_{\lambda_-(\alpha\bsA+\beta\bsB)}
\exp\Pa{-\mathrm{i}t}\dif
t=\exp\Pa{-\mathrm{i}(a_0\alpha+b_0\beta)}\frac{\sin\abs{\alpha\bsa+\beta\bsb}}{\abs{\alpha\bsa+\beta\bsb}},
\end{eqnarray*}
for
$$
\lambda_\pm(\alpha\bsA+\beta\bsB)=\alpha a_0+\beta b_0 \pm
\abs{\alpha\bsa+\beta\bsb}.
$$
Now that
\begin{eqnarray*}
f^{(2)}_{\langle\alpha\bsA+\beta\bsB\rangle}(r)
=\frac1{\lambda_2-\lambda_1}(H(r-\lambda_1)-H(r-\lambda_2)),
\end{eqnarray*}
therefore
\begin{eqnarray*}
f^{(2)}_{\langle\bsA\rangle,\langle\bsB\rangle}(r,s)=\frac1{(2\pi)^2}\int_{\real^2}\dif\alpha\dif\beta
\exp\Pa{\mathrm{i}((r-a_0)\alpha+(s-b_0)\beta)}\frac{\sin\abs{\alpha\bsa+\beta\bsb}}{\abs{\alpha\bsa+\beta\bsb}}.
\end{eqnarray*}
(i) If $\set{\bsa,\bsb}$ is linearly independent, then the following
matrix $\bsT_{\bsa,\bsb}$ is invertible, and thus
\begin{eqnarray*}
\Pa{\begin{array}{c}
      \tilde\alpha \\
      \tilde\beta
    \end{array}
} = \bsT^{\frac12}_{\bsa,\bsb}\Pa{\begin{array}{c}
      \alpha \\
      \beta
    \end{array}
},\quad\text{where }\bsT_{\bsa,\bsb} = \Pa{\begin{array}{cc}
                               \Inner{\bsa}{\bsa} & \Inner{\bsa}{\bsb} \\
                               \Inner{\bsa}{\bsb} & \Inner{\bsb}{\bsb}
                             \end{array}
}.
\end{eqnarray*}
Thus we see that
\begin{eqnarray*}
&&f^{(2)}_{\langle\bsA\rangle,\langle\bsB\rangle}(r,s)=\frac1{(2\pi)^2\sqrt{\det(\bsT_{\bsa,\bsb})}}\int_{\real^2}\dif\tilde\alpha\dif\tilde\beta
\exp\Pa{\mathrm{i}((\tilde r-\tilde a_0)\alpha+(\tilde s-\tilde
b_0)\tilde
\beta)}\frac{\sin\sqrt{\tilde\alpha^2+\tilde\beta^2}}{\sqrt{\tilde\alpha^2+\tilde\beta^2}}\\
&&=\frac1{(2\pi)^2\sqrt{\det(\bsT_{\bsa,\bsb})}}\int^\infty_0\dif t
\sin t \int^{2\pi}_0\dif\theta \exp\Pa{\mathrm{i}t((\tilde r-\tilde
a_0)\cos\theta+(\tilde s-\tilde b_0)\sin\theta)}\\
&&=\frac1{(2\pi)\sqrt{\det(\bsT_{\bsa,\bsb})}}\int^\infty_0\dif
t\sin t J_0\Pa{t\sqrt{(\tilde r-\tilde a_0)^2+(\tilde s-\tilde
b_0)^2}},
\end{eqnarray*}
where $J_0(z)$ is the so-called Bessel function of first kind,
defined by
\begin{eqnarray*}
J_0(z)=\frac1\pi\int^\pi_0\cos(z\cos\theta)\dif\theta.
\end{eqnarray*}
Therefore
\begin{eqnarray*}
f^{(2)}_{\langle\bsA\rangle,\langle\bsB\rangle}(r,s)
=\frac1{2\pi\sqrt{\det(\bsT_{\bsa,\bsb})}}\int^{+\infty}_0\dif t\sin
tJ_0\Pa{t\cdot\sqrt{(r-a_0,
s-b_0)\bsT_{\bsa,\bsb}^{-1}\Pa{\begin{array}{c}
      r-a_0 \\
      s-b_0
    \end{array}
}}},
\end{eqnarray*}
where
\begin{eqnarray*}
\int^{\infty}_0J_0(\lambda t)\sin(t)\dif t =
\frac1{\sqrt{1-\lambda^2}}H(1-\abs{\lambda}).
\end{eqnarray*}
Therefore
\begin{eqnarray*}
f^{(2)}_{\langle\bsA\rangle,\langle\bsB\rangle}(r,s)
=\frac{H(1-\omega_{\bsA,\bsB}(r,s))}{2\pi\sqrt{\det(\bsT_{\bsa,\bsb})(1-\omega^2_{\bsA,\bsB}(r,s))}},
\end{eqnarray*}
where
\begin{eqnarray}\label{eq:omegaAB}
\omega_{\bsA,\bsB}(r,s)=\sqrt{(r-a_0,
s-b_0)\bsT_{\bsa,\bsb}^{-1}\Pa{\begin{array}{c}
      r-a_0 \\
      s-b_0
    \end{array}
}}.
\end{eqnarray}
(ii) If $\set{\bsa,\bsb}$ is linearly dependent, without loss of
generality, let $\bsb=\kappa\cdot\bsa$ for some nonzero
$\kappa\neq0$, then
\begin{eqnarray*}
f^{(2)}_{\langle\bsA\rangle,\langle\bsB\rangle}(r,s)=\frac1{(2\pi)^2}\int_{\mR^2}\dif\alpha\dif\beta
\exp\Pa{\mathrm{i}((r-a_0)\alpha+(s-b_0)\beta)}\frac{\sin
(a\abs{\alpha+\beta\kappa})}{a\abs{\alpha+\beta\kappa}}.
\end{eqnarray*}
Here $a=\abs{\bsa}$. We perform the change of variables
$(\alpha,\beta)\to(\alpha',\beta')$, where
$\alpha'=\alpha+\kappa\beta$ and $\beta'=\beta$. We get its
Jacobian, given by
\begin{eqnarray*}
\det\Pa{\frac{\partial(\alpha',\beta')}{\partial(\alpha,\beta)}}=\abs{\begin{array}{cc}
                                                                        1 & \kappa \\
                                                                        0 &
                                                                        1
                                                                      \end{array}
}=1\neq0.
\end{eqnarray*}
Thus
\begin{eqnarray*}
f^{(2)}_{\langle\bsA\rangle,\langle\bsB\rangle}(r,s)&=&\frac1{(2\pi)^2}\iint\dif\alpha'\dif\beta'
\exp\Pa{\mathrm{i}((r-a_0)(\alpha'-\kappa\beta')+(s-b_0)\beta')}\frac{\sin
(a\abs{\alpha'})}{a\abs{\alpha'}}\\
&=&\frac1{2\pi}\int
\exp\Pa{\mathrm{i}((s-b_0)-\kappa(r-a_0))\beta'}\dif\beta'\times\frac1{2\pi}\int\dif\alpha'
\exp\Pa{\mathrm{i}(r-a_0)\alpha'}\frac{\sin(
a\abs{\alpha'})}{a\abs{\alpha'}}\\
&=&\delta((s-b_0)-\kappa(r-a_0))f^{(2)}_{\langle\bsA\rangle}(r).
\end{eqnarray*}

\begin{prop}\label{prop:expdis}
For a pair of qubit observables
$\bsA=a_0\I+\bsa\cdot\boldsymbol{\sigma}$ and
$\bsB=b_0\I+\bsb\cdot\boldsymbol{\sigma}$, (i) if $\set{\bsa,\bsb}$
is linearly independent, then the pdf of
$(\langle\bsA\rangle_\psi,\langle\bsB\rangle_\psi)$, where
$\psi\in\complex^2$ a Haar-distributed pure state, is given by
\begin{eqnarray*}
f^{(2)}_{\langle\bsA\rangle,\langle\bsB\rangle}(r,s)
=\frac{H(1-\omega_{\bsA,\bsB}(r,s))}{2\pi\sqrt{\det(\bsT_{\bsa,\bsb})(1-\omega^2_{\bsA,\bsB}(r,s))}}.
\end{eqnarray*}
(ii) If $\set{\bsa,\bsb}$ is linearly dependent, without loss of
generality, let $\bsb=\kappa\cdot\bsa$, then
\begin{eqnarray*}
f^{(2)}_{\langle\bsA\rangle,\langle\bsB\rangle}(r,s)=\delta((s-b_0)-\kappa(r-a_0))f^{(2)}_{\langle\bsA\rangle}(r).
\end{eqnarray*}
\end{prop}
From Proposition~\ref{prop:expdis}, we can directly infer the
results obtained in \cite{Gutkin2013,Gallay2012}.

We now turn to a pair of qubit observables
\begin{eqnarray}\label{AB}
\bsA =a_0\I+\bsa\cdot\boldsymbol{\sigma}, \quad \bsB
=b_0\I+\bsb\cdot\boldsymbol{\sigma} , \quad (a_0,\bsa),(b_0,\bsb)\in
\real^4,
\end{eqnarray}
whose uncertainty region
\begin{eqnarray}\label{2DUR}
\cU_{\Delta\bsA,\Delta\bsB}:=\Set{(\Delta_\psi \bsA,\Delta_\psi
\bsB)\in\real^2_+: \ket{\psi}\in \complex^2}
\end{eqnarray}
was proposed by Busch and Reardon-Smith \cite{Busch2019} in the
mixed state case. We consider the probability distribution density
\begin{eqnarray*}
f^{(2)}_{\Delta\bsA,\Delta\bsB}(x,y) := \int\delta(x-\Delta_\psi
\bsA)\delta(y-\Delta_\psi \bsB)\dif\mu (\psi),
\end{eqnarray*}
on the uncertainty region defined by Eq.~\eqref{2DUR}. Denote
\begin{eqnarray*}
\bsT_{\bsa,\bsb}:=\Pa{\begin{array}{cc}
                        \Inner{\bsa}{\bsa} & \Inner{\bsa}{\bsb} \\
                        \Inner{\bsb}{\bsa} & \Inner{\bsb}{\bsb}
                      \end{array}
}.
\end{eqnarray*}

\begin{thrm}\label{th:varAvarB}
The joint probability distribution density of the uncertainties
$(\Delta_\psi \bsA,\Delta_\psi \bsB)$ for a pair of qubit
observables defined by Eq.~\eqref{AB}, where $\psi$ is a
Haar-distributed random pure state on $\complex^2$, is given by
\begin{eqnarray*}
f^{(2)}_{\Delta\bsA,\Delta\bsB}(x,y) = \frac{2xy\sum_{j\in\set{\pm}}
f^{(2)}_{\langle\bsA\rangle,\langle\bsB\rangle}(r_+(x),s_j(y))}{\sqrt{(a^2-x^2)(b^2-y^2)}},
\end{eqnarray*}
where $a=\abs{\bsa}>0, b=\abs{\bsb}>0$,
$r_\pm(x)=a_0\pm\sqrt{a^2-x^2}, s_\pm(y)=b_0\pm\sqrt{b^2-y^2}$.
\end{thrm}

\begin{proof}
Note that in the proof of Theorem~\ref{th:vardis}, we have already
obtained that
\begin{eqnarray*}
\delta(x^2-\Delta^2_\psi\bsA) =
\delta(x^2-(r-\lambda_1(\bsA))(\lambda_2(\bsA)-r))) =\delta(f_x(r)),
\end{eqnarray*}
where $f_x(r):=x^2-(r-\lambda_1(\bsA))(\lambda_2(\bsA)-r)$.
Similarly,
\begin{eqnarray*}
\delta(y^2-\Delta^2_\psi\bsB) = \delta(g_y(s)),
\end{eqnarray*}
where $g_y(s)=y^2-(s-\lambda_1(\bsB))(\lambda_2(\bsB)-s)$.

Again, by using \eqref{eq:2nddelta}, we get that
\begin{eqnarray*}
f^{(2)}_{\Delta\bsA,\Delta\bsB}(x,y) &=&
4xy\int\delta(x^2-\Delta^2_\psi\bsA)\delta(y^2-\Delta^2_\psi\bsB)\dif\mu(\psi)\\
&=& 4 xy\iint\dif r\dif
sf^{(2)}_{\langle\bsA\rangle,\langle\bsB\rangle}(r,s)\delta(f_x(r))\delta(g_y(s)),
\end{eqnarray*}
where $f^{(2)}_{\langle\bsA\rangle,\langle\bsB\rangle}(r,s)$ is
determined by Proposition~\ref{prop:expdis}. Hence
\begin{eqnarray*}
\delta\Pa{f_x(r)}&=&\frac1{\abs{\partial_{r=r_+(x)}
f_x(r)}}\delta_{r_+(x)}+\frac1{\abs{\partial_{r=r_-(x)} f_x(r)}}\delta_{r_-(x)},\\
\delta\Pa{g_y(s)}&=&\frac1{\abs{\partial_{s=s_+(y)}
g_y(s)}}\delta_{s_+(y)}+\frac1{\abs{\partial_{s=s_-(y)}.
g_y(s)}}\delta_{s_-(y)}.
\end{eqnarray*}
From the above, we have already known that
\begin{eqnarray*}
\delta(f_x(r))\delta(g_y(s)) =
\frac{\delta_{(r_+,s_+)}+\delta_{(r_+,s_-)}+\delta_{(r_-,s_+)}+\delta_{(r_-,s_-)}}{4\sqrt{(a^2-x^2)(b^2-y^2)}}.
\end{eqnarray*}
Based on this observation, we get that
\begin{eqnarray*}
f^{(2)}_{\Delta\bsA,\Delta\bsB}(x,y) =
\frac{xy}{\sqrt{(a^2-x^2)(b^2-y^2)}}\sum_{i,j\in\set{\pm}}f_{\langle\bsA\rangle,\langle\bsB\rangle}(r_i(x),s_j(y)).
\end{eqnarray*}
It is easily checked that $\omega_{\bsA,\bsB}(\cdot,\cdot)$, defined
in \eqref{eq:omegaAB}, satisfies that
\begin{eqnarray*}
\omega_{\bsA,\bsB}(r_+(x),s_+(y))=\omega_{\bsA,\bsB}(r_-(x),s_-(y)),\quad
\omega_{\bsA,\bsB}(r_+(x),s_-(y))=\omega_{\bsA,\bsB}(r_-(x),s_+(y)).
\end{eqnarray*}
These lead to the fact that
\begin{eqnarray*}
\sum_{i,j\in\set{\pm}}f^{(2)}_{\langle\bsA\rangle,\langle\bsB\rangle}(r_i(x),s_j(y))=2\sum_{j\in\set{\pm}}f_{\langle\bsA\rangle,\langle\bsB\rangle}(r_+(x),s_j(y)).
\end{eqnarray*}
Therefore
\begin{eqnarray*}
f^{(2)}_{\Delta\bsA,\Delta\bsB}(x,y) =
\frac{2xy\sum_{j\in\set{\pm}}f^{(2)}_{\langle\bsA\rangle,\langle\bsB\rangle}(r_+(x),s_j(y))}{\sqrt{(a^2-x^2)(b^2-y^2)}}.
\end{eqnarray*}
We get the desired result.
\end{proof}

\subsection{The case for three qubit observables}\label{app:ABC}

We now turn to the case where there are three qubit observables
\begin{eqnarray}\label{ABC}
\bsA =a_0\I+\bsa\cdot\boldsymbol{\sigma}, \quad \bsB
=b_0\I+\bsb\cdot\boldsymbol{\sigma}, \quad \bsC
=c_0\I+\bsc\cdot\boldsymbol{\sigma} \quad
(a_0,\bsa),(b_0,\bsb),(c_0,\bsc)\in \real^4,
\end{eqnarray}
whose uncertainty region
\begin{eqnarray}\label{3DUR}
\cU_{\Delta\bsA,\Delta\bsB,\Delta\bsC}:=\Set{(\Delta_\psi
\bsA,\Delta_\psi \bsB,\Delta_\psi \bsC)\in\real^3_+: \ket{\psi}\in
\complex^2} .
\end{eqnarray}
We define the probability distribution density
\begin{eqnarray*}
f^{(2)}_{\Delta\bsA,\Delta\bsB,\Delta\bsC}(x,y,z) :=
\int\delta(x-\Delta_\psi \bsA)\delta(y-\Delta_\psi
\bsB)\delta(z-\Delta_\psi \bsC)\dif\mu (\psi),
\end{eqnarray*}
on the uncertainty region defined by Eq.~\eqref{3DUR}. Denote
\begin{eqnarray*}
\bsT_{\bsa,\bsb,\bsc}:=\Pa{\begin{array}{ccc}
                        \Inner{\bsa}{\bsa} & \Inner{\bsa}{\bsb} & \Inner{\bsa}{\bsc} \\
                        \Inner{\bsb}{\bsa} & \Inner{\bsb}{\bsb} & \Inner{\bsb}{\bsc} \\
                        \Inner{\bsc}{\bsa} & \Inner{\bsc}{\bsb} & \Inner{\bsc}{\bsc}
                      \end{array}
}.
\end{eqnarray*}
Again note that $\bsT_{\bsa,\bsb,\bsc}$ is also a semidefinite
positive matrix. We find that
$\rank(\bsT_{\bsa,\bsb,\bsc})\leqslant3$. There are three cases that
would be possible: $\rank(\bsT_{\bsa,\bsb,\bsc})=1,2,3$. Thus
$\bsT_{\bsa,\bsb,\bsc}$ is invertible (i.e.,
$\rank(\bsT_{\bsa,\bsb,\bsc})=3$) if and only if
$\set{\bsa,\bsb,\bsc}$ linearly independent. In such case, we write
\begin{eqnarray*}
\omega_{\bsA,\bsB,\bsC}(r,s,t):=\sqrt{(r-a_0,s-b_0,t-c_0)\bsT^{-1}_{\bsa,\bsb,\bsc}(r-a_0,s-b_0,t-c_0)^\t}.
\end{eqnarray*}
In order to calculate $f_{\Delta\bsA,\Delta\bsB,\Delta\bsC}$,
essentially we need to derive the joint probability distribution
density of
$(\langle\bsA\rangle_\psi,\langle\bsB\rangle_\psi,\langle\bsC\rangle_\psi)$,
which is defined by
\begin{eqnarray*}
f^{(2)}_{\langle\bsA\rangle,\langle\bsB\rangle,\langle\bsC\rangle}(r,s,t)
:=\int\delta(r-\langle\bsA\rangle_\psi)\delta(s-\langle\bsB\rangle_\psi)\delta(t-\langle\bsC\rangle_\psi)\dif\mu(\psi).
\end{eqnarray*}

We have the following result:

\begin{prop}\label{prop:ABC}
For three qubit observables, given by Eq.~\eqref{ABC}, (i) if
$\rank(\bsT_{\bsa,\bsb,\bsc})=3$, i.e., $\set{\bsa,\bsb,\bsc}$ is
linearly independent, then the joint probability distribution
density of
$(\langle\bsA\rangle_\psi,\langle\bsB\rangle_\psi,\langle\bsC\rangle_\psi)$,
where $\psi$ is a Haar-distributed random pure state on
$\complex^2$, is given by the following:
\begin{eqnarray}\label{eq:meanABC}
f^{(2)}_{\langle \bsA\rangle ,\langle \bsB\rangle , \langle
\bsC\rangle }(r,s,t)
=\frac1{4\pi\sqrt{\det(\bsT_{\bsa,\bsb,\bsc})}}\delta(1-\omega_{\bsA,\bsB,\bsC}(r,s,t)).
\end{eqnarray}
(ii) If $\rank(\bsT_{\bsa,\bsb,\bsc})=2$, i.e.,
$\set{\bsa,\bsb,\bsc}$ is linearly dependent, without loss of
generality, we assume that $\set{\bsa,\bsb}$ are linearly
independent and $\bsc=\kappa_{\bsa}\cdot\bsa+\kappa_{\bsb}\cdot\bsb$
for some $\kappa_{\bsa}$ and $\kappa_{\bsb}$ with
$\kappa_{\bsa}\kappa_{\bsb}\neq0$, then
\begin{eqnarray*}
f^{(2)}_{\langle\bsA\rangle,\langle\bsB\rangle,\langle\bsC\rangle}(r,s,t)=
\delta((t-c_0)-\kappa_{\bsa}(r-a_0)-\kappa_{\bsb}(s-b_0))f^{(2)}_{\langle\bsA\rangle,\langle\bsB\rangle}(r,s).
\end{eqnarray*}
(iii) If $\rank(\bsT_{\bsa,\bsb,\bsc})=1$, i.e.,
$\set{\bsa,\bsb,\bsc}$ is linearly dependent, without loss of
generality, we assume that $\bsa$ are linearly independent and
$\bsb=\kappa_{\bsb\bsa}\cdot\bsa,\bsc=\kappa_{\bsc\bsa}\cdot\bsa$
for some $\kappa_{\bsb\bsa}$ and $\kappa_{\bsc\bsa}$ with
$\kappa_{\bsb\bsa}\kappa_{\bsc\bsa}\neq0$, then
\begin{eqnarray*}
f^{(2)}_{\langle\bsA\rangle,\langle\bsB\rangle,\langle\bsC\rangle}(r,s,t)
=\delta((s-b_0)-\kappa_{\bsb\bsa}(r-a_0))\delta((t-c_0)-\kappa_{\bsc\bsa}(r-a_0))
f^{(2)}_{\langle\bsA\rangle}(r).
\end{eqnarray*}
\end{prop}

\begin{proof}
(i) If $\rank(\bsT_{\bsa,\bsb,\bsc})=3$, then
$\bsT_{\bsa,\bsb,\bsc}$ is \emph{invertible}. By using Bloch
representation,
$\proj{\psi}=\frac12(\I_2+\bsu\cdot\boldsymbol{\sigma})$, where
$\abs{\bsu}=1$. Then for
$(r,s,t)=(\langle\bsA\rangle_\psi,\langle\bsB\rangle_\psi,\langle\bsC\rangle_\psi)=(a_0+\Inner{\bsu}{\bsa},b_0+\Inner{\bsu}{\bsb},\Inner{\bsu}{\bsc})$,
we see that
\begin{eqnarray*}
(r-a_0,s-b_0,t-c_0) =
(\Inner{\bsu}{\bsa},\Inner{\bsu}{\bsb},\Inner{\bsu}{\bsc}).
\end{eqnarray*}
Denote $\bsQ:=(\bsa,\bsb,\bsc)$, which is a $3\times 3$ invertible
real matrix due to the fact that $\set{\bsa,\bsb,\bsc}$ is linearly
independent. Then $\bsT_{\bsa,\bsb,\bsc}=\bsQ^\t\bsQ$ and
$(r-a_0,s-b_0,t-c_0) =\bra{\bsu}\bsQ$, which means that
\begin{eqnarray*}
\omega_{\bsA,\bsB,\bsC}(r,s,t)=\sqrt{\bra{\bsu}\bsQ(\bsQ^\t\bsQ)^{-1}\bsQ^\t\ket{\bsu}}
= \abs{\bsu}=1.
\end{eqnarray*}
This tells us an interesting fact that
$(\langle\bsA\rangle_\psi,\langle\bsB\rangle_\psi,\langle\bsC\rangle_\psi)$
lies at the boundary surface of the ellipsoid
$\omega_{\bsA,\bsB,\bsC}(r,s,t)\leqslant1$, i.e.,
$\omega_{\bsA,\bsB,\bsC}(r,s,t)=1$. This indicates that the PDF of
$(\langle\bsA\rangle_\psi,\langle\bsB\rangle_\psi,\langle\bsC\rangle_\psi)$
satisfies that
\begin{eqnarray*}
f^{(2)}_{\langle\bsA\rangle,\langle\bsB\rangle,\langle\bsC\rangle}(r,s,t)\propto\delta(1-\omega_{\bsA,\bsB,\bsC}(r,s,t)).
\end{eqnarray*}
Next we calculate the following integral:
\begin{eqnarray*}
\int_{\real^3}\delta(1-\omega_{\bsA,\bsB,\bsC}(r,s,t))\dif r\dif
s\dif t=4\pi\sqrt{\det(\bsT_{\bsa,\bsb,\bsc})}.
\end{eqnarray*}
Apparently
\begin{eqnarray*}
\int_{\real^3}\delta(1-\omega_{\bsA,\bsB,\bsC}(r,s,t))\dif r\dif
s\dif t=
\int_{\real^3}\delta\Pa{1-\sqrt{\Innerm{\bsx}{\bsT^{-1}_{\bsa,\bsb,\bsc}}{\bsx}}}[\dif
\bsx].
\end{eqnarray*}
Here $\bsx=(r-a_0,s-b_0,t-c_0)$ and $[\dif\bsx]=\dif r\dif s\dif t$.
Indeed, by using spectral decomposition theorem for Hermitian
matrix, we get that there is orthogonal matrix $\bsO\in\O(3)$ such
that
$\bsT_{\bsa,\bsb,\bsc}=\bsO^\t\diag(\lambda_1,\lambda_2,\lambda_3)\bsO$
where $\lambda_k>0(k=1,2,3)$. Thus
\begin{eqnarray*}
\omega_{\bsA,\bsB,\bsC}(r,s,t)=\Innerm{\bsO\bsx}{\diag(\lambda^{-1}_1,\lambda^{-1}_2,\lambda^{-1}_3)}{\bsO\bsx}=\Innerm{\bsy}{\diag(\lambda^{-1}_1,\lambda^{-1}_2,\lambda^{-1}_3)}{\bsy}
\end{eqnarray*}
where $\bsy=\bsO\bsx$. Thus
\begin{eqnarray*}
\int_{\real^3}\delta(1-\omega_{\bsA,\bsB,\bsC}(r,s,t))\dif r\dif
s\dif t=
\int_{\real^3}\delta\Pa{1-\sqrt{\Innerm{\bsy}{\diag(\lambda^{-1}_1,\lambda^{-1}_2,\lambda^{-1}_3)}{\bsy}}}[\dif
\bsy].
\end{eqnarray*}
Let
$\bsz=\diag(\lambda^{-1/2}_1,\lambda^{-1/2}_2,\lambda^{-1/2}_3)\bsy$.
Then
$[\dif\bsz]=\frac1{\sqrt{\lambda_1\lambda_2\lambda_3}}[\dif\bsy]=\frac1{\sqrt{\det(\bsT_{\bsa,\bsb,\bsc})}}[\dif\bsy]$
and
\begin{eqnarray*}
\int_{\real^3}\delta(1-\omega_{\bsA,\bsB,\bsC}(r,s,t))\dif r\dif
s\dif t=
\sqrt{\det(\bsT_{\bsa,\bsb,\bsc})}\int_{\real^3}\delta\Pa{1-\abs{\bsz}}[\dif
\bsz]=4\pi \sqrt{\det(\bsT_{\bsa,\bsb,\bsc})}.
\end{eqnarray*}
Finally we get that
\begin{eqnarray*}
f^{(2)}_{\langle \bsA\rangle ,\langle \bsB\rangle , \langle
\bsC\rangle }(r,s,t) =
\frac1{4\pi\sqrt{\det(\bsT_{\bsa,\bsb,\bsc})}}\delta(1-\omega_{\bsA,\bsB,\bsC}(r,s,t)).
\end{eqnarray*}

(ii) If $\rank(\bsT_{\bsa,\bsb,\bsc})=2$, then
$\set{\bsa,\bsb,\bsc}$ is linearly dependent. Without loss of
generality, we assume that $\set{\bsa,\bsb}$ is independent. Now
$\bsc=\kappa_{\bsa}\bsa+\kappa_{\bsb}\bsb$ for some
$\kappa_{\bsa},\kappa_{\bsb}\in\real$ with
$\kappa_{\bsa}\kappa_{\bsb}\neq0$. Thus
\begin{eqnarray*}
t-c_0&=&\langle\bsC\rangle_\psi-c_0
=\Innerm{\psi}{\bsc\cdot\boldsymbol{\sigma}}{\psi}
=\kappa_{\bsa}\Innerm{\psi}{\bsa\cdot\boldsymbol{\sigma}}{\psi}+\kappa_{\bsb}\Innerm{\psi}{\bsb\cdot\boldsymbol{\sigma}}{\psi}\\
&&=\kappa_{\bsa}(r-a_0)+\kappa_{\bsb}(s-b_0).
\end{eqnarray*}
Therefore we get that
\begin{eqnarray*}
f^{(2)}_{\langle\bsA\rangle,\langle\bsB\rangle,\langle\bsC\rangle}(r,s,t)=
\delta((t-c_0)-\kappa_{\bsa}(r-a_0)-\kappa_{\bsb}(s-b_0))f^{(2)}_{\langle\bsA\rangle,\langle\bsB\rangle}(r,s).
\end{eqnarray*}

(iii) If $\rank(\bsT_{\bsa,\bsb,\bsc})=1$, then
$\set{\bsa,\bsb,\bsc}$ is linearly dependent. Without loss of
generality, we assume that $\bsa$ are linearly independent and
$\bsb=\kappa_{\bsb\bsa}\cdot\bsa,\bsc=\kappa_{\bsc\bsa}\cdot\bsa$
for some $\kappa_{\bsb\bsa}$ and $\kappa_{\bsc\bsa}$ with
$\kappa_{\bsb\bsa}\kappa_{\bsc\bsa}\neq0$. Then we get the desired
result by mimicking the proof in (ii).
\end{proof}

\begin{thrm}\label{th:ABC2}
The joint probability distribution density of $(\Delta_\psi
\bsA,\Delta_\psi\bsB,\Delta_\psi \bsC)$ for a triple of qubit
observables defined by Eq.~\eqref{ABC}, where $\ket{\psi}$ is a
Haar-distributed random pure state on $\complex^2$, is given by
\begin{eqnarray*}
f^{(2)}_{\Delta\bsA,\Delta\bsB,\Delta\bsC}(x,y,z)
=\frac{2xyz}{\sqrt{(a^2-x^2)(b^2-y^2)(c^2-z^2)}}\sum_{j,k\in\set{\pm}}
f^{(2)}_{\langle \bsA\rangle ,\langle \bsB\rangle,\langle
\bsC\rangle}(r_+(x),s_j(y),t_k(z)),
\end{eqnarray*}
where $f_{\langle \bsA\rangle,\langle \bsB\rangle, \langle
\bsC\rangle}(r,s,t)$ is the joint probability distribution density
of the expectation values $(\langle \bsA\rangle_\psi,\langle
\bsB\rangle_\psi,\langle \bsC\rangle_\psi)$, which is determined by
Eq.~\eqref{eq:meanABC} in Proposition~\ref{prop:ABC}; and
\begin{eqnarray*}
r_\pm(x):=a_0\pm\sqrt{a^2-x^2},\quad
s_\pm(y):=b_0\pm\sqrt{b^2-y^2},\quad t_\pm(z):=c_0\pm\sqrt{c^2-z^2}.
\end{eqnarray*}
\end{thrm}

\begin{proof}
Note that
\begin{eqnarray*}
f^{(2)}_{\Delta\bsA,\Delta\bsB,\Delta\bsC}(x,y,z)&=&\int\delta(x-\Delta_\psi
\bsA)\delta(y-\Delta_\psi
\bsB)\delta(z-\Delta_\psi\bsC)\dif\mu(\psi) \notag\\
&=& 8xyz\int\delta\Pa{x^2-\Delta^2_\psi
\bsA}\cdot\delta\Pa{y^2-\Delta^2_\psi
\bsB}\cdot\delta\Pa{z^2-\Delta^2_\psi\bsC}\dif\mu(\psi).
\end{eqnarray*}
Again using the method in the proof of Theorem 1, we have already
obtained that
\begin{eqnarray*}
\delta\Pa{x^2-\Delta^2_\psi \bsA}\cdot\delta\Pa{y^2-\Delta^2_\psi
\bsB}\cdot\delta\Pa{z^2-\Delta^2_\psi
\bsC}=\delta(f_x(r))\delta(g_y(s))\delta(h_z(t)),
\end{eqnarray*}
where
\begin{eqnarray*}
f_x(r)&:=&x^2-(r-\lambda_1(\bsA))(\lambda_2(\bsA)-r),\\
g_y(s)&:=&y^2-(s-\lambda_1(\bsB))(\lambda_2(\bsB)-s),\\
h_z(t)&:=&z^2-(t-\lambda_1(\bsC))(\lambda_2(\bsC)-t).
\end{eqnarray*}
Then
\begin{eqnarray*}
f^{(2)}_{\Delta\bsA,\Delta\bsB,\Delta\bsC}(x,y,z)&=&8xyz\iiint
\delta(f_x(r))\delta(g_y(s))\delta(h_z(t))f^{(2)}_{\langle\bsA\rangle,\langle\bsB\rangle,\langle\bsC\rangle}(r,s,t)\dif
r\dif s\dif t.
\end{eqnarray*}
Furthermore we have
\begin{eqnarray*}
\delta(f_x(r))\delta(g_y(s))\delta(h_z(t)) =
\frac{\sum_{i,j,k\in\set{\pm}}\delta_{(r_i(x),s_j(y),t_k(z))}}{8\sqrt{(a^2-x^2)(b^2-y^2)(c^2-z^2)}}.
\end{eqnarray*}
Based on this observation, we get that
\begin{eqnarray*}
f^{(2)}_{\Delta\bsA,\Delta\bsB,\Delta\bsC}(x,y,z)&=&\frac{xyz}{\sqrt{(a^2-x^2)(b^2-y^2)(c^2-z^2)}}\sum_{i,j,k\in\set{\pm}}\Inner{\delta_{(r_i(x),s_j(y),t_k(z))}}{f^{(2)}_{\langle\bsA\rangle,\langle\bsB\rangle,\langle\bsC\rangle}}.
\end{eqnarray*}
Thus
\begin{eqnarray*}
f^{(2)}_{\Delta\bsA,\Delta\bsB,\Delta\bsC}(x,y,z) =
\frac{xyz\sum_{i,j,k\in\set{\pm}}f^{(2)}_{\langle\bsA\rangle,\langle\bsB\rangle,\langle\bsC\rangle}(r_i(x),s_j(y),t_k(z))}{\sqrt{(a^2-x^2)(b^2-y^2)(c^2-z^2)}}.
\end{eqnarray*}
It is easily seen that
\begin{eqnarray*}
f^{(2)}_{\langle\bsA\rangle,\langle\bsB\rangle,\langle\bsC\rangle}(r_+(x),s_+(y),t_+(z))=f^{(2)}_{\langle\bsA\rangle,\langle\bsB\rangle,\langle\bsC\rangle}(r_-(x),s_-(y),t_-(z)),\\
f^{(2)}_{\langle\bsA\rangle,\langle\bsB\rangle,\langle\bsC\rangle}(r_+(x),s_+(y),t_-(z))=f^{(2)}_{\langle\bsA\rangle,\langle\bsB\rangle,\langle\bsC\rangle}(r_-(x),s_-(y),t_+(z)),\\
f^{(2)}_{\langle\bsA\rangle,\langle\bsB\rangle,\langle\bsC\rangle}(r_+(x),s_-(y),t_+(z))=f^{(2)}_{\langle\bsA\rangle,\langle\bsB\rangle,\langle\bsC\rangle}(r_-(x),s_+(y),t_-(z)),\\
f^{(2)}_{\langle\bsA\rangle,\langle\bsB\rangle,\langle\bsC\rangle}(r_+(x),s_-(y),t_-(z))=f^{(2)}_{\langle\bsA\rangle,\langle\bsB\rangle,\langle\bsC\rangle}(r_-(x),s_+(y),t_+(z)).
\end{eqnarray*}
From these observations, we can reduce the above expression to the
following:
\begin{eqnarray*}
f^{(2)}_{\Delta\bsA,\Delta\bsB,\Delta\bsC}(x,y,z) =
\frac{2xyz\sum_{j,k\in\set{\pm}}f^{(2)}_{\langle\bsA\rangle,\langle\bsB\rangle,\langle\bsC\rangle}(r_+(x),s_j(y),t_k(z))}{\sqrt{(a^2-x^2)(b^2-y^2)(c^2-z^2)}}.
\end{eqnarray*}
The desired result is obtained.
\end{proof}

Note that the PDFs of uncertainties of multiple qubit observables
(more than three) will be reduced into the three situations above,
as shown in \cite{Zhang2021preprint}. Here we will omit it here.

\section{PDF of uncertainty of a single qudit observable}

Assume $\bsA$ is a non-degenerate positive matrix with eigenvalues
$\lambda_k(\bsA)=a_k(k=1,\ldots,d)$ with $a_d>\cdots>a_1$. Denote by
$V_d(\bsa)=\prod_{1\leqslant i<j\leqslant d}(a_j-a_i)$. Due to the
following relation
$(\Delta_\psi\bsA)^2=\Innerm{\psi}{\bsA^2}{\psi}-\Innerm{\psi}{\bsA}{\psi}^2$,
i.e., the variance of $\bsA$ is the function of
$r=\Innerm{\psi}{\bsA}{\psi}$ and $s=\Innerm{\psi}{\bsA^2}{\psi}$,
where $\ket{\psi}$ is a Haar-distributed pure state. Thus firstly we
derive the joint PDF of
$(\Innerm{\psi}{\bsA}{\psi},\Innerm{\psi}{\bsA^2}{\psi})$, defined
by
\begin{eqnarray*}
f^{(d)}_{\langle\bsA\rangle,\langle\bsA^2\rangle}(r,s) := \int
\delta(r-\Innerm{\psi}{\bsA}{\psi})\delta(s-\Innerm{\psi}{\bsA^2}{\psi})\dif\mu(\psi).
\end{eqnarray*}
By performing Laplace transformation $(r,s)\to(\alpha,\beta)$ to
$f^{(d)}_{\langle\bsA\rangle,\langle\bsA^2\rangle}(r,s)$, we get
that
\begin{eqnarray*}
\sL\Pa{f^{(d)}_{\langle\bsA\rangle,\langle\bsA^2\rangle}}(\alpha,\beta)
=\int
\exp\Pa{-\Innerm{\psi}{\alpha\bsA+\beta\bsA^2}{\psi}}\dif\mu(\psi)=
\int
\exp\Pa{-z}f^{(d)}_{\langle\alpha\bsA+\beta\bsA^2\rangle}(z)\dif z,
\end{eqnarray*}
where $f^{(d)}_{\langle\alpha\bsA+\beta\bsA^2\rangle}(z)$ is
determined by Proposition~\ref{prop:expectn}:
\begin{eqnarray*}
f^{(d)}_{\langle\alpha\bsA+\beta\bsA^2\rangle}(z)=(-1)^{d-1}(d-1)\sum^d_{i=1}\frac{(z-(\alpha
a_i+\beta a^2_i))^{d-2}}{\prod_{j\in\hat i}(a_i-a_j)(\alpha+\beta
(a_i+a_j))}H(z-(\alpha a_i+\beta a^2_i)).
\end{eqnarray*}
Theoretically, we can calculate the above integral for any finite
natural number $d$, but instead, we will focus on the case where
$d=3,4$, we use \textit{Mathematica} to do this tedious job. By
simplifying the results obtained via the Laplace
transformation/inverse Laplace transformation in
\textit{Mathematica}, we get the following results without details.

\begin{thrm}\label{th:qu3it}
For a given qutrit observable $\bsA$, acting on $\complex^3$, with
their eigenvalues $a_1<a_2<a_3$, the joint pdf of
$(\langle\bsA\rangle_\psi,\langle\bsA^2\rangle_\psi)$, where
$\ket{\psi}\in\mC^3$, is given by
\begin{eqnarray*}
f^{(3)}_{\langle\bsA\rangle,\langle\bsA^2\rangle}(r,s)
=\frac{\Gamma(3)}{V_3(\bsa)},
\end{eqnarray*}
on $D=D_1\cup D_2$ for
\begin{eqnarray*}
D_1&=&\Set{(r,s): a_1\leqslant r\leqslant a_2,
(a_1+a_2)r-a_1a_2\leqslant s\leqslant(a_1+a_3)r-a_1a_3},\\
D_2&=&\Set{(r,s): a_2\leqslant r\leqslant a_3,
(a_2+a_3)r-a_2a_3\leqslant s\leqslant(a_1+a_3)r-a_1a_3};
\end{eqnarray*}
and $f^{(3)}_{\langle\bsA\rangle,\langle\bsA^2\rangle}(r,s)=0$
otherwise. Thus
\begin{eqnarray*}
f^{(3)}_{\langle\bsA\rangle,\Delta\bsA}(r,x) =
\frac{2{\Gamma(3)}}{V_3(\bsa)}x\quad(\forall (r,x)\in R),
\end{eqnarray*}
where $R=R_1\cup R_2$ with
\begin{eqnarray*}
R_1&=&\Set{(r,x): a_1\leqslant r\leqslant a_2,
\sqrt{(a_2-r)(r-a_1)}\leqslant x\leqslant
\sqrt{(a_3-r)(r-a_1)}},\\
R_2&=&\Set{(r,x): a_2\leqslant r\leqslant a_3,
\sqrt{(a_3-r)(r-a_2)}\leqslant x\leqslant \sqrt{(a_3-r)(r-a_1)}},
\end{eqnarray*}
and $f^{(3)}_{\langle\bsA\rangle,\Delta\bsA}(r,x) =0$ otherwise.
Moreover, we get that
\begin{eqnarray*}
f^{(3)}_{\Delta\bsA}(x) =
\frac{4\Gamma(3)}{V_3(\bsa)}x\Pa{\chi_{\Br{0,\tfrac{a_3-a_1}2}}(x)\varepsilon_{31}(x)-\chi_{\Br{0,\tfrac{a_3-a_2}2}}(x)\varepsilon_{32}(x)-\chi_{\Br{0,\tfrac{a_2-a_1}2}}(x)\varepsilon_{21}(x)},
\end{eqnarray*}
where $x\in\Br{0,\frac{a_3-a_1}2}$ and $\chi_S(x)$ is the indicator
of the set $S$, i.e., $\chi_S(x)=1$ if $x\in S$, and $\chi_S(x)=0$
if $x\notin S$;
\begin{eqnarray}\label{eq:notationdelta}
\varepsilon_{ij}(x):=\sqrt{\Pa{\frac{a_i-a_j}2}^2- x^2}.
\end{eqnarray}
\end{thrm}

\begin{thrm}
For a given qudit observable $\bsA$, acting on $\complex^4$, with
their eigenvalues $a_1<a_2<a_3<a_4$, the joint pdf of
$(\langle\bsA\rangle_\psi,\langle\bsA^2\rangle_\psi)$, where
$\ket{\psi}\in\mC^4$, is given by
\begin{eqnarray*}
f^{(4)}_{\langle\bsA\rangle,\langle\bsA^2\rangle}(r,s)
=\frac{\Gamma(4)}{V_4(\bsa)}g(r,s),
\end{eqnarray*}
where $g(r,s)$ is defined by the following way: via
$\varphi_{i,j}(r):=(a_i+a_j)r-a_ia_j$,
\begin{eqnarray*}
g(r,s) =
\begin{cases}
(a_4-a_3)(s-\varphi_{1,2}(r)),&\text{if }(r,s)\in D^{(1)}_1\cup D^{(2)}_2;\\
(a_4-a_1)(s-\varphi_{2,3}(r)),&\text{if }(r,s)\in D^{(2)}_1\cup D^{(2)}_4;\\
(a_2-a_1)(s-\varphi_{3,4}(r)),&\text{if }(r,s)\in D^{(2)}_5\cup D^{(3)}_1;\\
(a_2-a_3)(s-\varphi_{1,4}(r)),&\text{if }(r,s)\in D^{(1)}_2\cup D^{(2)}_3\cup D^{(2)}_6\cup D^{(3)}_2;\\
\end{cases}
\end{eqnarray*}
on $D=\bigcup^2_{i=1}D^{(1)}_i\cup \bigcup^6_{j=1} D^{(2)}_j\cup
\bigcup^2_{k=1}D^{(3)}_k$ for
\begin{eqnarray*}
D^{(1)}_1&=&\Set{(r,s): a_1\leqslant r\leqslant a_2,
\varphi_{1,2}(r)\leqslant s\leqslant\varphi_{1,3}(r)},\\
D^{(1)}_2&=&\Set{(r,s): a_1\leqslant r\leqslant a_2,
\varphi_{1,3}(r)\leqslant s\leqslant\varphi_{1,4}(r)},
\end{eqnarray*}
and
\begin{eqnarray*}
D^{(2)}_1&=& \Set{(r,s): a_2\leqslant r\leqslant
\frac{a_2a_4-a_1a_3}{a_2+a_4-a_1-a_3},
\varphi_{2,3}(r)\leqslant s\leqslant\varphi_{2,4}(r)},\\
D^{(2)}_2&=& \Set{(r,s): a_2\leqslant r\leqslant
\frac{a_2a_4-a_1a_3}{a_2+a_4-a_1-a_3},
\varphi_{2,4}(r)\leqslant s\leqslant\varphi_{1,3}(r)},\\
D^{(2)}_3&=& \Set{(r,s): a_2\leqslant r\leqslant
\frac{a_2a_4-a_1a_3}{a_2+a_4-a_1-a_3},
\varphi_{1,3}(r)\leqslant s\leqslant\varphi_{1,4}(r)},\\
D^{(2)}_4&=& \Set{(r,s):
\frac{a_2a_4-a_1a_3}{a_2+a_4-a_1-a_3}\leqslant r\leqslant a_3,
\varphi_{2,3}(r)\leqslant s\leqslant\varphi_{1,3}(r)},\\
D^{(2)}_5&=& \Set{(r,s):
\frac{a_2a_4-a_1a_3}{a_2+a_4-a_1-a_3}\leqslant r\leqslant a_3,
\varphi_{1,3}(r)\leqslant s\leqslant\varphi_{2,4}(r)},\\
D^{(2)}_6&=& \Set{(r,s):
\frac{a_2a_4-a_1a_3}{a_2+a_4-a_1-a_3}\leqslant r\leqslant a_3,
\varphi_{2,4}(r)\leqslant s\leqslant\varphi_{1,4}(r)},
\end{eqnarray*}
and
\begin{eqnarray*}
D^{(3)}_1&=&\Set{(r,s): a_3\leqslant r\leqslant a_4,
\varphi_{3,4}(r)\leqslant s\leqslant\varphi_{2,4}(r)},\\
D^{(3)}_2&=&\Set{(r,s): a_3\leqslant r\leqslant a_4,
\varphi_{2,4}(r)\leqslant s\leqslant\varphi_{1,4}(r)},
\end{eqnarray*}
and $f^{(4)}_{\langle\bsA\rangle,\langle\bsA^2\rangle}(r,s)=0$
otherwise.
\end{thrm}

\begin{remark}
Denote by $D_1=\bigcup^2_{i=1}D^{(1)}_i,
D_2=\bigcup^6_{i=1}D^{(2)}_j$, and $D_3=\bigcup^2_{k=1}D^{(3)}_k$,
respectively. Thus
\begin{eqnarray*}
D_1 &=& \Set{(r,s): a_1\leqslant r\leqslant a_2,
\varphi_{1,2}(r)\leqslant s\leqslant\varphi_{1,4}(r)},\\
D_2 &=& \Set{(r,s): a_2\leqslant r\leqslant a_3,
\varphi_{2,3}(r)\leqslant s\leqslant\varphi_{1,4}(r)},\\
D_3 &=& \Set{(r,s): a_3\leqslant r\leqslant a_4,
\varphi_{3,4}(r)\leqslant s\leqslant\varphi_{1,4}(r)}.
\end{eqnarray*}
This implies that the support of
$f^{(4)}_{\langle\bsA\rangle,\langle\bsA^2\rangle}$ is just $D_1\cup
D_2\cup D_3$, i.e.,
$\supp(f^{(4)}_{\langle\bsA\rangle,\langle\bsA^2\rangle})=D_1\cup
D_2\cup D_3$.
\end{remark}

\begin{cor}\label{cor:DeltaA4}
For a given qudit observable $\bsA$, acting on $\complex^4$, with
their eigenvalues $a_1<a_2<a_3<a_4$, the joint pdf of
$(\langle\bsA\rangle_\psi,\langle\bsA^2\rangle_\psi)$, where
$\ket{\psi}\in\mC^4$, is given by
\begin{eqnarray*}
f^{(4)}_{\langle\bsA\rangle,\Delta\bsA}(r,x)
=\frac{2\Gamma(4)}{V_4(\bsa)}xg(r,x^2+r^2).
\end{eqnarray*}
Moreover, via $\varphi_{i,j}(r):=(a_i+a_j)r-a_ia_j$,
\begin{eqnarray*}
f^{(4)}_{\langle\bsA\rangle,\Delta\bsA}(r,x)
=\frac{2\Gamma(4)}{V_4(\bsa)}x
\begin{cases}
(a_4-a_3)(x^2+r^2-\varphi_{1,2}(r)),&\text{if }(r,s)\in R^{(1)}_1\cup R^{(2)}_2;\\
(a_4-a_1)(x^2+r^2-\varphi_{2,3}(r)),&\text{if }(r,s)\in R^{(2)}_1\cup R^{(2)}_4;\\
(a_2-a_1)(x^2+r^2-\varphi_{3,4}(r)),&\text{if }(r,s)\in R^{(2)}_5\cup R^{(3)}_1;\\
(a_2-a_3)(x^2+r^2-\varphi_{1,4}(r)),&\text{if }(r,s)\in R^{(1)}_2\cup R^{(2)}_3\cup R^{(2)}_6\cup R^{(3)}_2;\\
\end{cases}
\end{eqnarray*}
on $R=\bigcup^2_{i=1}R^{(1)}_i\cup \bigcup^6_{j=1} R^{(2)}_j\cup
\bigcup^2_{k=1}R^{(3)}_k$ for
\begin{eqnarray*}
R^{(1)}_1&=&\Set{(r,x): a_1\leqslant r\leqslant a_2,
\sqrt{(a_2-r)(r-a_1)}\leqslant x\leqslant\sqrt{(a_3-r)(r-a_1)}},\\
R^{(1)}_2&=&\Set{(r,x): a_1\leqslant r\leqslant a_2,
\sqrt{(a_3-r)(r-a_1)}\leqslant x\leqslant\sqrt{(a_4-r)(r-a_1)}},
\end{eqnarray*}
and
\begin{eqnarray*}
R^{(2)}_1&=& \Set{(r,x): a_2\leqslant r\leqslant
\frac{a_2a_4-a_1a_3}{a_2+a_4-a_1-a_3},
\sqrt{(a_3-r)(r-a_2)}\leqslant x\leqslant\sqrt{(a_4-r)(r-a_2)}},\\
R^{(2)}_2&=& \Set{(r,x): a_2\leqslant r\leqslant
\frac{a_2a_4-a_1a_3}{a_2+a_4-a_1-a_3},
\sqrt{(a_4-r)(r-a_2)}\leqslant x\leqslant\sqrt{(a_3-r)(r-a_1)}},\\
R^{(2)}_3&=& \Set{(r,x): a_2\leqslant r\leqslant
\frac{a_2a_4-a_1a_3}{a_2+a_4-a_1-a_3},
\sqrt{(a_3-r)(r-a_1)}\leqslant x\leqslant\sqrt{(a_4-r)(r-a_1)}},\\
R^{(2)}_4&=& \Set{(r,x):
\frac{a_2a_4-a_1a_3}{a_2+a_4-a_1-a_3}\leqslant r\leqslant a_3,
\sqrt{(a_3-r)(r-a_2)}\leqslant x\leqslant\sqrt{(a_3-r)(r-a_1)}},\\
R^{(2)}_5&=& \Set{(r,x):
\frac{a_2a_4-a_1a_3}{a_2+a_4-a_1-a_3}\leqslant r\leqslant a_3,
\sqrt{(a_3-r)(r-a_1)}\leqslant x\leqslant\sqrt{(a_4-r)(r-a_2)}},\\
R^{(2)}_6&=& \Set{(r,x):
\frac{a_2a_4-a_1a_3}{a_2+a_4-a_1-a_3}\leqslant r\leqslant a_3,
\sqrt{(a_4-r)(r-a_2)}\leqslant x\leqslant\sqrt{(a_4-r)(r-a_1)}},
\end{eqnarray*}
and
\begin{eqnarray*}
R^{(3)}_1&=&\Set{(r,x): a_3\leqslant r\leqslant a_4,
\sqrt{(a_4-r)(r-a_3)}\leqslant x\leqslant\sqrt{(a_4-r)(r-a_2)}},\\
R^{(3)}_2&=&\Set{(r,x): a_3\leqslant r\leqslant a_4,
\sqrt{(a_4-r)(r-a_2)}\leqslant x\leqslant\sqrt{(a_4-r)(r-a_1)}},
\end{eqnarray*}
and $f^{(4)}_{\langle\bsA\rangle,\Delta\bsA}(r,x)=0$ otherwise.
Moreover $f^{(4)}_{\Delta\bsA}(x)$ can be identified as
\begin{eqnarray*}
f^{(4)}_{\Delta\bsA}(x)&=&\frac{4^2}{V_4(\bsa)}x\Big[(a_4-a_3)\chi_{\Br{0,\frac{a_2-a_1}2}}(x)\varepsilon^3_{12}(x)+(a_4-a_1)\chi_{\Br{0,\frac{a_3-a_2}2}}(x)\varepsilon^3_{23}(x)\notag\\
&&~~~~~~~~~~~~~-(a_4-a_2)\chi_{\Br{0,\frac{a_3-a_1}2}}(x)\varepsilon^3_{13}(x)+(a_2-a_1)\chi_{\Br{0,\frac{a_4-a_3}2}}(x)\varepsilon^3_{34}(x)\notag\\
&&~~~~~~~~~~~~~-(a_3-a_1)\chi_{\Br{0,\frac{a_4-a_2}2}}(x)\varepsilon^3_{24}(x)+(a_3-a_2)\chi_{\Br{0,\frac{a_4-a_1}2}}(x)\varepsilon^3_{14}(x)\Big].
\end{eqnarray*}
Here the meanings of the notations $\chi$ and $\varepsilon_{ij}$ can
be found in Theorem~\ref{th:qu3it}.
\end{cor}

\begin{remark}
Denote by $R_1=\bigcup^2_{i=1}R^{(1)}_i,
R_2=\bigcup^6_{i=1}R^{(2)}_j$, and $R_3=\bigcup^2_{k=1}R^{(3)}_k$,
respectively. Thus
\begin{eqnarray*}
R_1 &=& \Set{(r,x): a_1\leqslant r\leqslant a_2,
\sqrt{(a_2-r)(r-a_1)}\leqslant x\leqslant\sqrt{(a_4-r)(r-a_1)}},\\
R_2 &=& \Set{(r,x): a_2\leqslant r\leqslant a_3,
\sqrt{(a_3-r)(r-a_2)}\leqslant x\leqslant\sqrt{(a_4-r)(r-a_1)}},\\
R_3 &=& \Set{(r,x): a_3\leqslant r\leqslant a_4,
\sqrt{(a_4-r)(r-a_3)}\leqslant x\leqslant\sqrt{(a_4-r)(r-a_1)}}.
\end{eqnarray*}
This implies that the support of
$f^{(4)}_{\langle\bsA\rangle,\Delta\bsA}$ is just $R_1\cup R_2\cup
R_3$, i.e., $\supp(f^{(4)}_{\langle\bsA\rangle,\Delta\bsA})=R_1\cup
R_2\cup R_3$.

We draw the plot of the support of
$f^{(d)}_{\langle\bsA\rangle,\Delta\bsA}(r,x)$, where $d=3,4$, as
below:
\begin{figure}[ht]
\subfigure[The support of
$f^{(3)}_{\langle\bsA\rangle,\Delta\bsA}(r,x)$, where
$\lambda(\bsA)=(1,3,9)$] {\begin{minipage}[b]{.49\linewidth}
\includegraphics[width=1\textwidth]{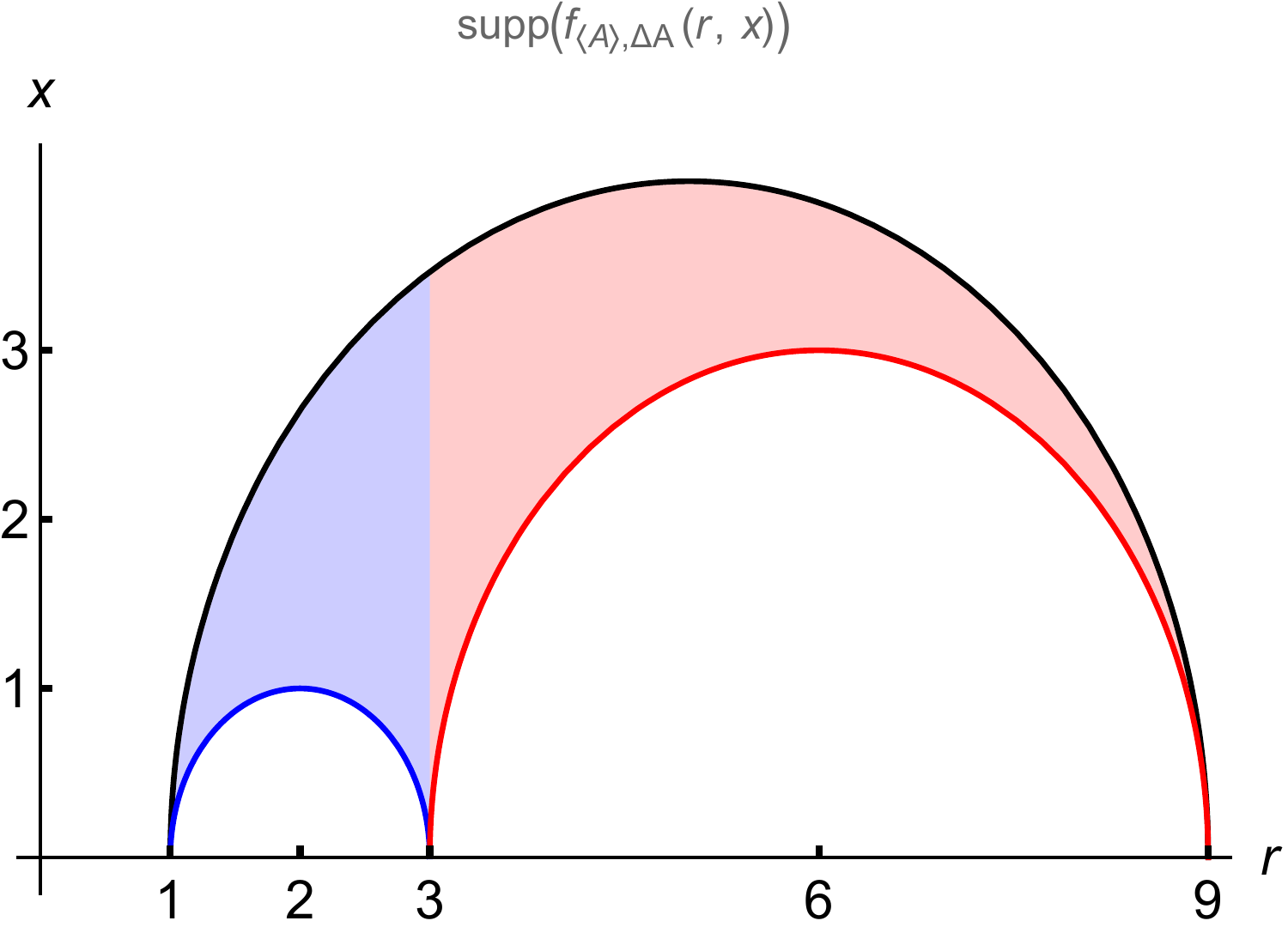}
\end{minipage}}
\subfigure[The support of
$f^{(4)}_{\langle\bsA\rangle,\Delta\bsA}(r,x)$, where
$\lambda(\bsA)=(1,3,9,27)$] {\begin{minipage}[b]{.49\linewidth}
\includegraphics[width=1\textwidth]{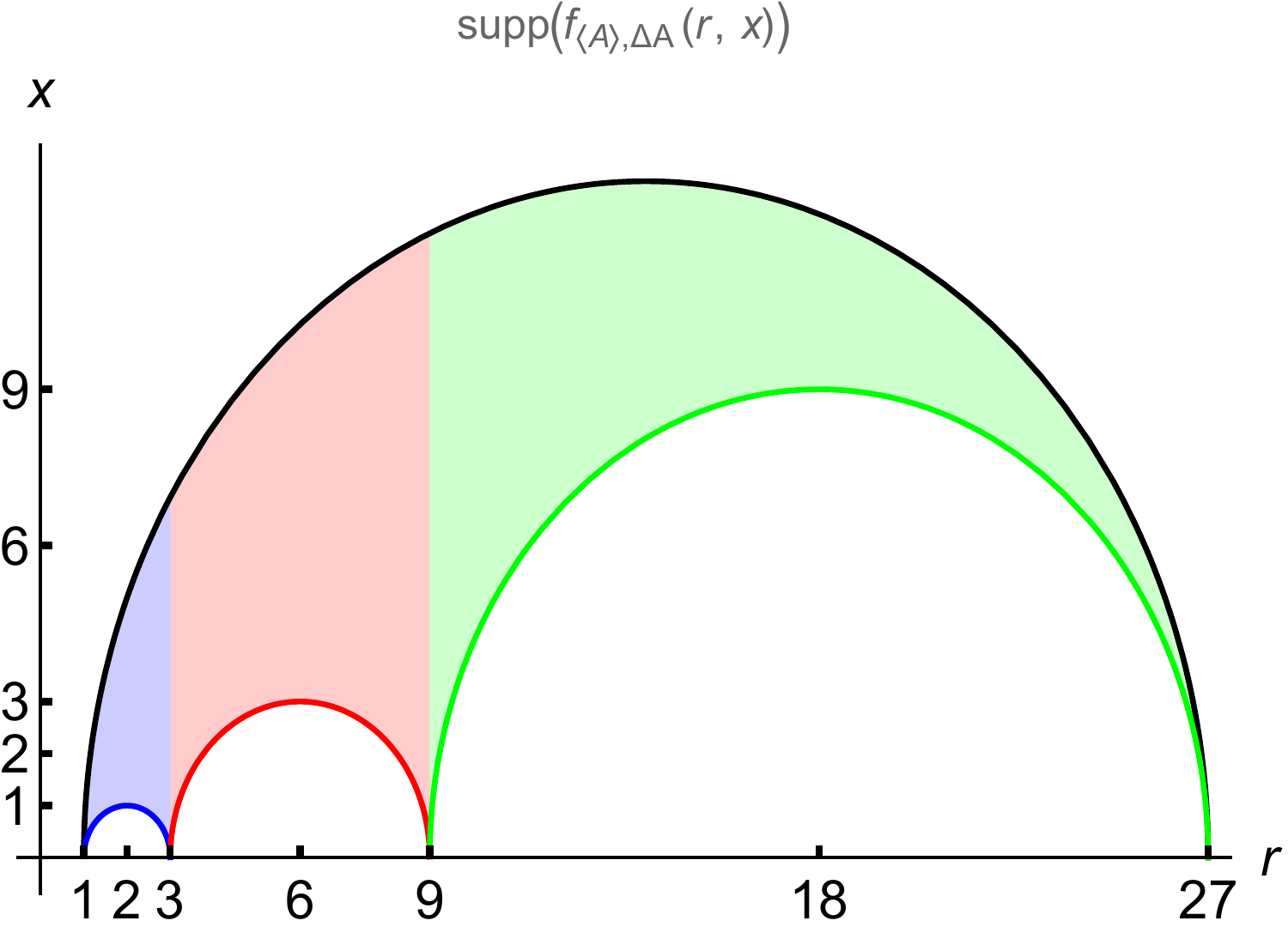}
\end{minipage}}
\caption{Plots of the supports of
$f^{(d)}_{\langle\bsA\rangle,\Delta\bsA}(r,x)$ for qudit observables
$\bsA$.}\label{fig:supp3&supp4}
\end{figure}
\end{remark}

Based on the above Corollary~\ref{cor:DeltaA4}, we can derive
$f^{(4)}_{\Delta\bsA}(x)$ just like to do similarly for
$f^{(3)}_{\Delta\bsA}(x)$. As an illustration, we will present a
specific example where the eigenvalues of $\bsA$ are given by
$\lambda(\bsA)=(1,3,9,27)$. In fact, this approach can goes for any
qudit observable with much computational complexity. In addition,
deriving the joint PDF of two uncertainties
$(\Delta\bsA,\Delta\bsB)$ of two qudit observables $\bsA$ and
$\bsB$, $f^{(d)}_{\Delta\bsA,\Delta\bsB}$, is very complicated. This
is not the goal of the present paper.

\begin{exam}
For a qudit observable $\bsA$, acting on $\complex^4$, with the
eigenvalues $\lambda(\bsA)=(1,3,9,27)$. Still employing the notation
in \eqref{eq:notationdelta} here:
\begin{eqnarray*}
\varepsilon_{12}(x)=\sqrt{1-x^2},\quad
\varepsilon_{13}(x)=\sqrt{4^2-x^2},\quad
\varepsilon_{14}(x)=\sqrt{13^2-x^2},\\
\varepsilon_{23}(x)=\sqrt{3^2-x^2},\quad
\varepsilon_{24}(x)=\sqrt{12^2-x^2},\quad
\varepsilon_{34}(x)=\sqrt{9^2-x^2}.
\end{eqnarray*}
Then from Corollary~\ref{cor:DeltaA4}, using marginal integral, we
can derive the PDF of $\Delta_\psi\bsA$ that \\
(i) If $x\in[0,1]$, then
\begin{eqnarray*}
f^{(4)}_{\Delta\bsA}(x) =\frac{x}{33696}\Pa{
9\varepsilon^3_{12}(x)+13\varepsilon^3_{23}(x)-12\varepsilon^3_{13}(x)+\varepsilon^3_{34}(x)-4\varepsilon^3_{24}(x)+3\varepsilon^3_{14}(x)}.
\end{eqnarray*}
(ii) If $x\in[1,3]$, then
\begin{eqnarray*}
f^{(4)}_{\Delta\bsA}(x)
=\frac{x}{33696}\Pa{13\varepsilon^3_{23}(x)-12\varepsilon^3_{13}(x)+\varepsilon^3_{34}(x)-4\varepsilon^3_{24}(x)+3\varepsilon^3_{14}(x)}.
\end{eqnarray*}
(iii) If $x\in[3,4]$, then
\begin{eqnarray*}
f^{(4)}_{\Delta\bsA}(x)
=\frac{x}{33696}\Pa{-12\varepsilon^3_{13}(x)+\varepsilon^3_{34}(x)-4\varepsilon^3_{24}(x)+3\varepsilon^3_{14}(x)}.
\end{eqnarray*}
(iv) If $x\in[4,9]$, then
\begin{eqnarray*}
f^{(4)}_{\Delta\bsA}(x)
=\frac{x}{33696}\Pa{\varepsilon^3_{34}(x)-4\varepsilon^3_{24}(x)+3\varepsilon^3_{14}(x)}.
\end{eqnarray*}
(v) If $x\in[9,12]$, then
\begin{eqnarray*}
f^{(4)}_{\Delta\bsA}(x)
=\frac{x}{33696}\Pa{-4\varepsilon^3_{24}(x)+3\varepsilon^3_{14}(x)}.
\end{eqnarray*}
(vi) If $x\in[12,13]$, then
\begin{eqnarray*}
f^{(4)}_{\Delta\bsA}(x) =\frac{x}{11232}\varepsilon^3_{14}(x).
\end{eqnarray*}
For illustrations, we draw the plot of $f^{(d)}_{\Delta\bsA}(x)$ for
qudit observables $\bsA$, where $d=3,4$, as below:
\begin{figure}[ht]
\subfigure[The PDF $f^{(3)}_{\Delta\bsA}(x)$, where
$\lambda(\bsA)=(1,3,9)$] {\begin{minipage}[b]{.49\linewidth}
\includegraphics[width=1\textwidth]{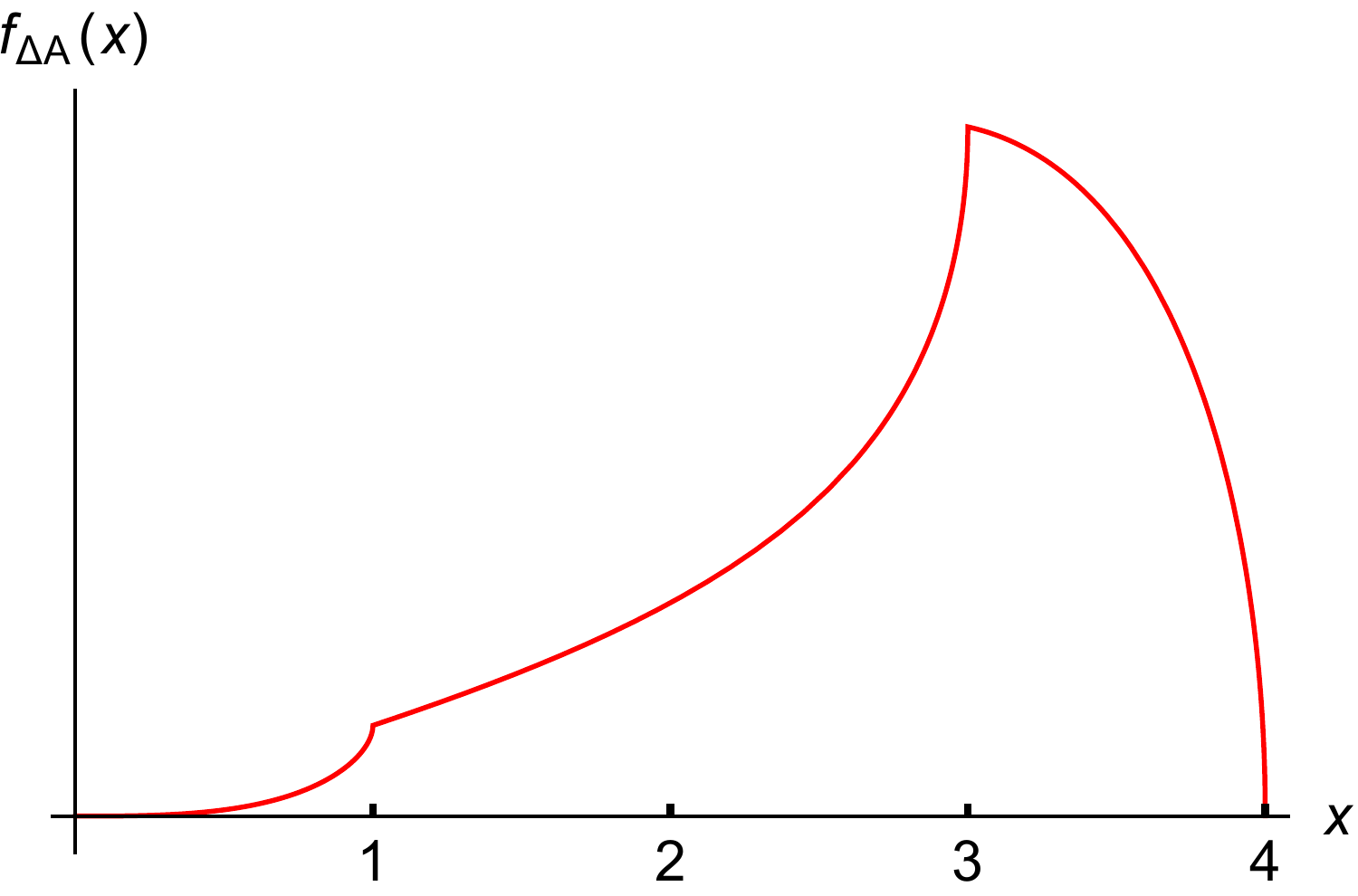}
\end{minipage}}
\subfigure[The PDF $f^{(4)}_{\Delta\bsA}(x)$, where
$\lambda(\bsA)=(1,3,9,27)$] {\begin{minipage}[b]{.49\linewidth}
\includegraphics[width=1\textwidth]{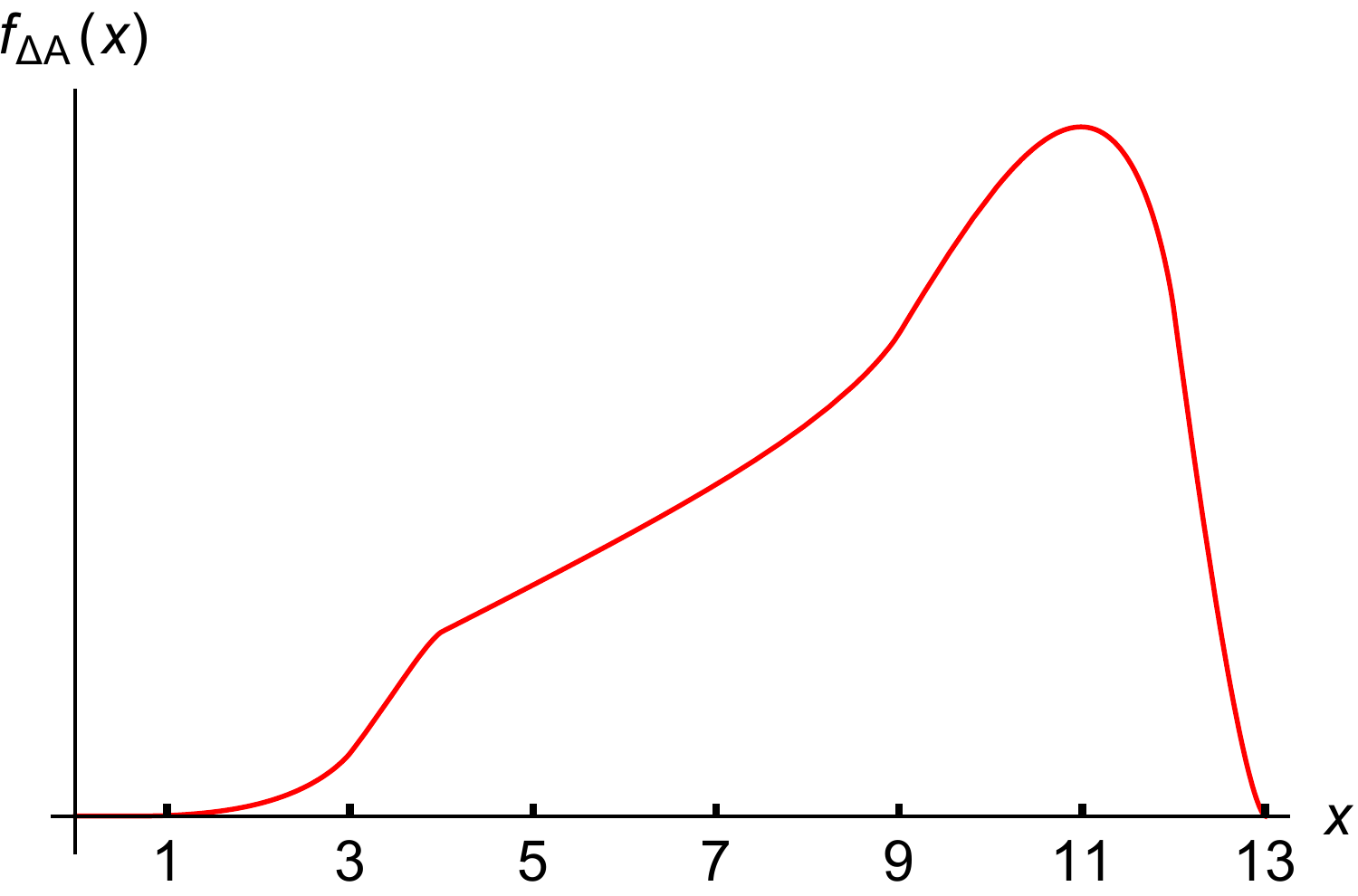}
\end{minipage}}
\caption{Plots of the PDFs $f^{(d)}_{\Delta\bsA}(x)$ for qudit
observables $\bsA$.}\label{fig:varpdf3&varpdf4}
\end{figure}
\end{exam}
In the last, we will identify the supports of
$f^{(d)}_{\langle\bsA\rangle,\langle\bsA^2\rangle}(r,s)$ and
$f^{(d)}_{\langle\bsA\rangle,\Delta\bsA}(r,x)$, where
\begin{eqnarray*}
f^{(d)}_{\langle\bsA\rangle,\langle\bsA^2\rangle}(r,s) &=&
\int\delta(r-\langle\bsA\rangle_\psi)\delta(s-\langle\bsA^2\rangle_\psi)\dif\mu(\psi),\\
f^{(d)}_{\langle\bsA\rangle,\Delta\bsA}(r,x) &=& 2x
f^{(d)}_{\langle\bsA\rangle,\langle\bsA^2\rangle}\Pa{r,r^2+x^2}.
\end{eqnarray*}

\begin{thrm}
For a qudit observable $\bsA$, acting on $\complex^d(d>1)$, with
eigenvalues $\lambda(\bsA)=(a_1,\ldots,a_d)$, where
$a_1<\cdots<a_d$, the supports of the PDFs of
$f^{(d)}_{\langle\bsA\rangle,\langle\bsA^2\rangle}(r,s)$ and
$f^{(d)}_{\langle\bsA\rangle,\Delta\bsA}(r,x)$, respectively, given
by the following:
\begin{eqnarray*}
\supp\Pa{f^{(d)}_{\langle\bsA\rangle,\langle\bsA^2\rangle}} =
\bigcup^{d-1}_{k=1}F_{k,k+1},
\end{eqnarray*}
where
\begin{eqnarray*}
F_{k,k+1} := \Set{(r,s): a_k\leqslant r\leqslant
a_{k+1},\varphi_{k,k+1}(r)\leqslant s\leqslant \varphi_{1,d}(r)};
\end{eqnarray*}
\begin{eqnarray*}
\supp\Pa{f^{(d)}_{\langle\bsA\rangle,\Delta\bsA}} =
\bigcup^{d-1}_{k=1} V_{k,k+1},
\end{eqnarray*}
where
\begin{eqnarray*}
V_{k,k+1} := \Set{(r,x): a_k\leqslant r\leqslant
a_{k+1},\sqrt{(a_{k+1}-r)(r-a_k)}\leqslant x\leqslant
\sqrt{(a_d-r)(r-a_1)}}.
\end{eqnarray*}
\end{thrm}

\begin{proof}
Without loss of generality, we assume that
$\bsA=\diag(a_1,\ldots,a_d)$ with $a_1<\cdots<a_d$. Let
$\ket{\psi}=(\psi_1,\ldots,\psi_d)^\t\in\complex^d$ be a pure state
and $r_k=\abs{\psi_k}^2$. Thus $(r_1,\ldots,r_d)$ is a
$d$-diemsional probability vector. Then
\begin{eqnarray*}
r=\langle\bsA\rangle = \sum^d_{k=1}a_kr_k\quad\text{and}\quad
s=\langle\bsA^2\rangle = \sum^d_{k=1}a^2_kr_k.
\end{eqnarray*}
Thus, for each $k\in\set{1,\ldots,d-1}$,
\begin{eqnarray*}
&&s-\varphi_{k,k+1}(r) = \sum^d_{i=1}a^2_ir_i -
(a_k+a_{k+1})\sum^d_{i=1}a_ir_i+a_ka_{k+1}\\
&&= \sum^d_{i=1}a^2_ir_i -
(a_k+a_{k+1})\sum^d_{i=1}a_ir_i+a_ka_{k+1}\sum^d_{i=1}r_i\\
&&=\sum^d_{i=1}(a_i-a_k)(a_i-a_{k+1})r_i.
\end{eqnarray*}
Note that $a_1<\cdots<a_d$ and $r_i\geqslant0$ for each
$i=1,\ldots,d$. We see that, when $i=k,k+1$,
$(a_i-a_k)(a_i-a_{k+1})r_i=0$, and
$(a_i-a_k)(a_i-a_{k+1})r_i\geqslant0$ otherwise. This means that
$s\geqslant \varphi_{k,k+1}(r)$; and this inequality is saturated if
$(r_k,r_{k+1})=(t,1-t)$ for $t\in(0,1)$ and $r_i=0$ for $i\neq
k,k+1$. Similarly, we can easily get that $s\leqslant
\varphi_{1,d}(r)$. Denote
\begin{eqnarray*} F_{k,k+1} = \Set{(r,s):
a_k\leqslant r\leqslant a_{k+1},\varphi_{k,k+1}(r)\leqslant
s\leqslant \varphi_{1,d}(r)};
\end{eqnarray*}
Hence
\begin{eqnarray*}
\supp\Pa{f^{(d)}_{\langle\bsA\rangle,\langle\bsA^2\rangle}} =
\bigcup^{d-1}_{k=1}F_{k,k+1}.
\end{eqnarray*}
Now, for $x=\Delta_\psi\bsA$, we see that $s=r^2+x^2$. By employing
the support of $f^{(d)}_{\langle\bsA\rangle,\langle\bsA^2\rangle}$,
we can derive the support of
$f^{(d)}_{\langle\bsA\rangle,\Delta\bsA}(r,x)$ as follows: Denote
\begin{eqnarray*}
V_{k,k+1} = \Set{(r,x): a_k\leqslant r\leqslant
a_{k+1},\sqrt{(a_{k+1}-r)(r-a_k)}\leqslant x\leqslant
\sqrt{(a_d-r)(r-a_1)}},
\end{eqnarray*}
then
\begin{eqnarray*}
\varphi_{k,k+1}(r)\leqslant s=r^2+x^2\leqslant
\varphi_{1,d}(r)\Longleftrightarrow (r,x)\in V_{k,k+1}.
\end{eqnarray*}
Therefore the support of
$f^{(d)}_{\langle\bsA\rangle,\Delta\bsA}(r,x)$ is given by
\begin{eqnarray*}
\supp\Pa{f^{(d)}_{\langle\bsA\rangle,\Delta\bsA}} =
\bigcup^{d-1}_{k=1} V_{k,k+1}.
\end{eqnarray*}
This completes the proof.
\end{proof}

For the joint PDF of uncertainties of multiple qudit observables
acting on $\mC^d(d\geqslant3)$, say, a pair of qudit observables
$(\bsA,\bsB)$, deriving the joint PDF
$f^{(d)}_{\Delta\bsA,\Delta\bsB}(x,y)$ is very complicated because
there is much difficulty in calculating the Laplace
transformation/inverse Laplace transformation of
$f^{(d)}_{\Delta\bsA,\Delta\bsB}(x,y)$. The reason is that we still
cannot figure out what the relationship among
$\lambda_k(\alpha\bsA+\beta\bsB), \lambda_k(\bsA)$, and
$\lambda_k(\bsB)$ is for varied $(\alpha,\beta)\in\mR^2$. A fresh
method to do this is expected to discover in the future.

\section{Discussion and concluding remarks}

Recall that the support $\supp(f)$ of a function
$f$ is given by the closure of the subset of preimage for
which $f$ does not vanish. From Theorem~\ref{th:vardis}, we see that the
support of $f^{(2)}_{\Delta\bsA}$ is the closed
interval $[0,\abs{\bsa}]$. This is in consistent with the fact that
$\Delta_\psi\bsA\in[0,v(\bsA)]$, where
$v(\bsA):=\tfrac12(\lambda_{\max}(\bsA)-\lambda_{\min}(\bsA))$ and
$\ket{\psi}$ is any pure state.

From Proposition~\ref{prop:expdis} and Theorem~\ref{th:varAvarB}, we
can infer that, for $d=2$, each element in
$\cU^{(\text{p})}_{\Delta\bsA,\Delta\bsB}$ is just the solution of
the following inequality:
\begin{eqnarray*}
\abs{\bsb}^2x^2+\abs{\bsa}^2y^2+2\abs{\Inner{\bsa}{\bsb}}\sqrt{(\abs{\bsa}^2-x^2)(\abs{\bsb}^2-y^2)}\geqslant
\abs{\bsa}^2\abs{\bsb}^2+\abs{\Inner{\bsa}{\bsb}}^2,
\end{eqnarray*}
which is exactly the one we obtained in \cite{Zhang2021preprint} for
mixed states. This indicates that, in the qubit situation, we have
that
\begin{prop}
For a pair of qubit observables $(\bsA,\bsB)$ acting on $\mC^2$, it
holds that
\begin{eqnarray*}
\cU^{(\mathrm{p})}_{\Delta\bsA,\Delta\bsB} =
\cU^{(\mathrm{m})}_{\Delta\bsA,\Delta\bsB}.
\end{eqnarray*}
\end{prop}
One may wonder if this identity holds for general $d\geqslant2$,
as the variance of $\bsA$ with respect to a mixed state
can always be decomposed as a convex combination of some variances of
$\bsA$ associated to pure states, see Eq.~\eqref{eq:vardecom}.

For multiple qudit observables $\bsA_k$ $(k=1,\ldots,n)$ acting on
$\complex^d$, comparing the set
$\cU^{(\text{p})}_{\Delta\bsA_1,\ldots,\Delta\bsA_n}$ with the set $
\cU^{(\text{m})}_{\Delta\bsA_1,\ldots,\Delta\bsA_n}$ is an
interesting problem. Unfortunately, our Theorem~\ref{th:ABC2},
together the result obtained in \cite{Zhang2021preprint}, indicates
that
$\cU^{(\text{m})}_{\Delta\bsA_1,\Delta\bsA_2,\Delta\bsA_3}
=\cU^{(\text{p})}_{\Delta\bsA_1,\Delta\bsA_2,\Delta\bsA_3}$
does not hold in general. In fact,
$\partial\cU^{(\text{m})}_{\Delta\bsA_1,\Delta\bsA_2,\Delta\bsA_3}
=\cU^{(\text{p})}_{\Delta\bsA_1,\Delta\bsA_2,\Delta\bsA_3}$
in the qubit situations, that is, the boundary surface of
$\cU^{(\text{m})}_{\Delta\bsA_1,\Delta\bsA_2,\Delta\bsA_3}$ is just
$\cU^{(\text{p})}_{\Delta\bsA_1,\Delta\bsA_2,\Delta\bsA_3}$ in the
qubit situations. This also indicates that the
following inclusion is proper in general for multiple observables,
\begin{eqnarray*}
\cU^{(\text{p})}_{\Delta\bsA_1,\ldots,\Delta\bsA_n}\subsetneq\cU^{(\text{m})}_{\Delta\bsA_1,\ldots,\Delta\bsA_n}.
\end{eqnarray*}
Based on this, two extreme cases:
$\cU^{(\text{p})}_{\Delta\bsA_1,\ldots,\Delta\bsA_n}=\cU^{(\text{m})}_{\Delta\bsA_1,\ldots,\Delta\bsA_n}$
or
$\cU^{(\text{p})}_{\Delta\bsA_1,\ldots,\Delta\bsA_n}=\partial\cU^{(\text{m})}_{\Delta\bsA_1,\ldots,\Delta\bsA_n}$
should be characterized.

In addition, we also see that once we obtain the uncertainty
regions for observables $\bsA_k$, we can infer additive uncertainty
relations such as
\begin{eqnarray*}
\sum^n_{k=1}(\Delta_\rho\bsA_k)^2\geqslant
\min_{\rho\in\density{\complex^d}}\sum^n_{k=1}(\Delta_\rho\bsA_k)^2
=
\min\Set{\sum^n_{k=1}x^2_k:(x_1,\ldots,x_n)\in\cU^{(\text{m})}_{\Delta\bsA_1,\ldots,\Delta\bsA_n}},
\end{eqnarray*}
or
\begin{eqnarray*}
\sum^n_{k=1}\Delta_\rho\bsA_k\geqslant
\min_{\rho\in\density{\complex^d}}\sum^n_{k=1}\Delta_\rho\bsA_k =
\min\Set{\sum^n_{k=1}x_k:(x_1,\ldots,x_n)\in\cU^{(\text{m})}_{\Delta\bsA_1,\ldots,\Delta\bsA_n}}.
\end{eqnarray*}
Analogous optimal problems can also be considered for
$\cU^{(\text{p})}_{\Delta\bsA_1,\ldots,\Delta\bsA_n}$. These results
can be used to detect entanglement
\cite{Guhne2004prl,Schwonnek2017prl}. The current results and the
results in \cite{Zhang2021preprint} together give the complete
solutions to the uncertainty region and uncertainty relations for
qubit observables.

We hope the results obtained in the present paper can shed new
lights on the related problems in quantum information theory. Our
approach may be applied to the study on PDFs in higher dimensional
spaces. It would be also interesting to apply PDFs to measurement
and/or quantum channel uncertainty relations.

\subsubsection*{Acknowledgments}
This work is supported by the NSF of China under Grant Nos.
11971140, 12075159, and 12171044, Beijing Natural Science Foundation
(Z190005), the Academician Innovation Platform of Hainan Province,
and Academy for Multidisciplinary Studies, Capital Normal
University. LZ is also funded by Natural Science Foundations of
Hubei Province Grant No. 2020CFB538.



\begin{thebibliography}{999}


\bibitem{Heisenberg1927}
W. Heisenberg, {\em \"{U}ber den anschaulichen Inhalt der
quantentheoretischen Kinematik und Mechanik}, Zeitschrift f\"{u}r
Physik (in German)
\href{https://doi.org/10.1007/BF01397280}{\textbf{43}(3-4): 172-198
(1927).}
%

\bibitem{Dammeier2015njp}
L. Dammeier, R. Schwonneck, and R.F. Werner, {\em Uncertainty
relations for angular momentum}, \njp
\href{https://doi.org/10.1088/1367-2630/17/9/093046}{{\bf17}, 093046
(2015).}
%

\bibitem{Li2015}
J.L. Li and C.F. Qiao,  {\em Reformulating the Quantum Uncertainty
Relation}, Sci. Rep.
\href{https://doi.org/10.1038/srep12708}{{\bf5}, 12708 (2015).}
%

\bibitem{Guise2018pra}
H. de Guise, L. Maccone, B.C. Sanders, and N. Shukla, {\em
State-independent uncertainty relations}, \pra
\href{https://doi.org/10.1103/PhysRevA.98.042121}{{\bf98}, 042121
(2018).}
%

\bibitem{Giorda2019pra}
P. Giorda, L. Maccone, and A. Riccardi, {\em State-independent
uncertainty relations from eigenvalue minimization}, \pra
\href{https://doi.org/10.1103/PhysRevA.99.052121}{{\bf99}, 052121
(2019).}
%

\bibitem{Xiao2019pra}
Y. Xiao, C. Guo, F. Meng, N. Jing, and M-H. Yung, {\em
Incompatibility of observables as state-independent bound of
uncertainty relations}, \pra
\href{https://doi.org/10.1103/PhysRevA.100.032118}{{\bf100}, 032118
(2019).}
%

\bibitem{Sponar2020pra}
S. Sponar, A. Danner, K. Obigane, S. Hack, and Y. Hasegawa, {\em
Experimental test of tight state-independent preparation uncertainty
relations for qubits}, \pra
\href{https://doi.org/10.1103/PhysRevA.102.042204}{{\bf102}, 042204
(2020).}
%


\bibitem{Seife2005}
C. Seife, {\em Do Deeper Principles Underlie Quantum Uncertainty and
Nonlocality?}, Science
\href{https://doi.org/10.1126/science.309.5731.98}{{\bf309} (5731),
98 (2005).}
%

\bibitem{Hofman2003pra}
H.F. Hofmann and S. Takeuchi, {\em Violation of local uncertainty
relations as a signature of entanglement}, \pra
\href{https://doi.org/10.1103/PhysRevA.68.032103}{{\bf68}, 032103
(2003).}
%


\bibitem{Guhne2004prl}
O. G\"{u}hne, {\em Characterizing Entanglement via Uncertainty
Relations}, \prl
\href{https://doi.org/10.1103/PhysRevLett.92.117903}{{\bf92}, 117903
(2004).}
%


\bibitem{Guhne2009pra}
O. G\"{u}hne and G. T\'{o}th, {\em Entanglement detection}, Phys.
Rep. \href{https://doi.org/10.1016/j.physrep.2009.02.004}{{\bf474},
1 (2009).}
%

\bibitem{Schwonnek2017prl}
R. Schwonnek, L. Dammeier, and R.F. Werner, {\em State-Independent
Uncertainty Relations and Entanglement Detection in Noisy Systems},
\prl \href{https://doi.org/10.1103/PhysRevLett.119.170404}{{\bf
119}, 170404 (2017).}
%

\bibitem{Qian2018qip}
C. Qian, J-L. Li, C-F. Qiao, {\em State-independent uncertainty
relations and entanglement detection}, Quant. Inf. Process
\href{https://doi.org/10.1007/s11128-018-1855-4}{{\bf17}: 84
(2018).}
%

\bibitem{Zhao2019prl}
Y-Y. Zhao, G.Y. Xiang, X.M. Hu, B.H. Liu, C.F. Li, G.C. Guo, R.
Schwonnek, R. Wolf, {\em Entanglement Detection by Violations of
Noisy Uncertainty Relations:A Proof of Principle}, \prl
\href{https://doi.org/10.1103/PhysRevLett.122.220401}{{\bf122},
220401 (2019).}
%

\bibitem{Oppenheim2010}
J. Oppenheim and S. Wehner, {\em The Uncertainty Principle
Determines the Nonlocality of Quantum Mechanics}, \sci~
\href{https://doi.org/10.1126/science.1192065}{{\bf330}, 1072-1074
(2010).}
%


\bibitem{Kennard1927}
E.H. Kennard, {\em Zur Quantenmechanik einfacher Bewegungstypen}, Z.
f\"{u}r Phys. \href{https://doi.org/10.1007/BF01391200}{{\bf44}(4),
326-352 (1927).}
%

\bibitem{Weyl1928}
H. Weyl, Gruppentheorie und Quantenmechanik (Leipzig, Hirzel, 1928).

\bibitem{Robertson1929}
H.P. Robertson, {\em The Uncertainty Principle}, Phys. Rev.
\href{https://doi.org/10.1103/PhysRev.34.163}{{\bf34}, 163-164
(1929).}
%

\bibitem{Schrodinger1930}
E. Schr\"{o}dinger, {\em Zum Heisenbergschen Unscharfeprinzip},
Sitzungsberichte der Preussischen Akademie der Wissenschaften,
Physikalisch-mathematische Klasse {\bf14} 296-303 (1930).
%


\bibitem{Busch2019}
P. Busch and O. Reardon-Smith, {\em On Quantum Uncertainty Relations
and Uncertainty Regions},
\href{https://arxiv.org/abs/1901.03695}{arXiv:1901.03695}
%

\bibitem{Zhang2018}
L. Zhang and J. Wang, {\em Average of Uncertainty Product for
Bounded Observables}, Open Systems \& Information Dynamics
\href{https://doi.org/10.1142/S1230161218500087}{{\bf25}(2), 1850008
(2018).}


\bibitem{Maccone2014}
L. Maccone and A.K. Pati, {\em Stronger uncertainty relations for
all incompatible observables}, \prl
\href{https://doi.org/10.1103/PhysRevLett.113.260401}{{\bf113},
260401 (2014).}

\bibitem{Hastings2009}
M.B. Hastings, {\em Superadditivity of communication capacity using
entangled inputs}, \natphy
\href{https://doi.org/10.1038/nphys1224}{{\bf5}(4), 255-257 (2009).}


\bibitem{Christandl2014}
M. Christandl, B. Doran, S. Kousidis, and M. Walter, {\em Eigenvalue
Distributions of Reduced Density Matrices}, \cmp
\href{https://doi.org/10.1007/s00220-014-2144-4}{{\bf332}(1), 1-52
(2014).}


\bibitem{Dartois2020}
S. Dartois, L. Lionni, and I. Nechita, {\em The joint distribution
of the marginals of multipartite random quantum states}, Random
Matrices: Theory and Applications
\href{https://doi.org/10.1142/s2010326320500100}{{\bf9}(3), 2050010
(2020).}


\bibitem{Zhang2018pla}
L. Zhang, J. Wang, and Z.H. Chen, {\em Spectral density of mixtures
of random density matrices for qubits}, \pla
\href{https://doi.org/10.1016/j.physleta.2018.04.018}{{\bf382}(23),
1516-1523 (2018).}


\bibitem{Zhang2019jpa}
L. Zhang, Y.X. Jiang, and J.D. Wu, {\em Duistermaat-Heckman measure
and the mixture of quantum states}, \jpa
\href{https://doi.org/10.1088/1751-8121/ab5297}{{\bf52}, 495203
(2019).}


\bibitem{Venuti2013}
L.C. Venuti, P. Zanardi, {\em Probability density of quantum
expectation values}, \pla
\href{https://doi.org/10.1016/j.physleta.2013.05.041}{\textbf{377}(31-33),
1854-1861 (2013).}

\bibitem{Zhang2021preprint}
L. Zhang, S. Luo, S-M. Fei, J. Wu, {\em Uncertainty Regions of
Observables and State-Independent Uncertainty Relations}, Quantum
Inf Process
\href{https://doi.org/10.1007/s11128-021-03303-w}{{\bf20},
357(2021).}

\bibitem{Hoskins2009}
R.F. Hoskins, Delta function, Elsevier (2009).

\bibitem{lz2020ijtp}
L. Zhang, {\em Dirac Delta Function of Matrix Argument}, Int. J.
Theor. Phys.
\href{https://doi.org/10.1007/s10773-020-04598-8}{{\bf60},
2445-2472(2021).}


\bibitem{Zuber2020}
M. Bauer and J-B. Zuber,
{\em On Products of Delta Distributions and
Resultants}, \href{https://doi.org/10.3842/SIGMA.2020.083}{SIGMA
{\bf16}, 083 (2020).}


\bibitem{Petz2012}
D. Petz and G. T\'{o}th,
{\em Matrix variances with projections},
Acta Sci. Math. (Szeged) {\bf78}, 683-688 (2012).
%

\bibitem{Bhatia2000}
R. Bhatia and C. Davis, {\em A better bound on the variance}, Amer.
Math. Month  \href{
https://doi.org/10.1080/00029890.2000.12005203}{{\bf107(4)}, 353-357
(2000).}

\bibitem{zhang2018qip}
L. Zhang, Z. Ma, Z. Chen, and S-M. Fei, {\em Coherence generating
power of unitary transformations via probabilistic average}, Quant
Inf Process
\href{https://doi.org/10.1007/s11128-018-1928-4}{{\bf17}: 186
(2018).}

\bibitem{Gutkin2013}
E. Gutkin, K. \.{Z}yczkowski, {\em Joint numerical ranges, quantum
maps, and joint numerical shadows}, Lin. Alg. Appl.
\href{https://doi.org/10.1016/j.laa.2012.10.043}{{\bf438}, 2394-2404
(2013).}

\bibitem{Gallay2012}
T. Gallay, D. Serre, {\em The numerical measure of a complex
matrix}, Commun. Pure Appl. Math.
\href{https://doi.org/10.1002/cpa.20374}{{\bf65}, 287-336(2012).}

\end{thebibliography}
\end{document}